\documentclass[12pt]{article}
\usepackage[colorlinks=true,linkcolor=blue,citecolor=Red]{hyperref}
\usepackage{breakcites}

\usepackage{mathtools}

\usepackage{subfigure}

\usepackage{algorithm}
\usepackage{algpseudocode}

\usepackage{latexsym}
\usepackage{graphics}
\usepackage{amsmath}
\usepackage[algo2e]{algorithm2e}
\usepackage{xspace}
\usepackage{amssymb}
\usepackage{psfrag}
\usepackage{epsfig}
\usepackage{amsthm}
\usepackage{url}
\usepackage{pst-all}
\usepackage{bbm}

\usepackage[title]{appendix}

\usepackage{color}

\definecolor{Red}{rgb}{1,0,0}
\definecolor{Blue}{rgb}{0,0,1}
\definecolor{Olive}{rgb}{0.41,0.55,0.13}
\definecolor{Yarok}{rgb}{0,0.5,0}
\definecolor{Green}{rgb}{0,1,0}
\definecolor{MGreen}{rgb}{0,0.8,0}
\definecolor{DGreen}{rgb}{0,0.55,0}
\definecolor{Yellow}{rgb}{1,1,0}
\definecolor{Cyan}{rgb}{0,1,1}
\definecolor{Magenta}{rgb}{1,0,1}
\definecolor{Orange}{rgb}{1,.5,0}
\definecolor{Violet}{rgb}{.5,0,.5}
\definecolor{Purple}{rgb}{.75,0,.25}
\definecolor{Brown}{rgb}{.75,.5,.25}
\definecolor{Grey}{rgb}{.5,.5,.5}

\newcommand{\bs}{\boldsymbol{\sigma}}

\usepackage{bbm}
\newcommand{\ind}{\mathbbm{1}}

\newcommand{\R}{\mathbb{R}}

\newcommand{\N}{\mathbb{N}}
\newcommand{\ip}[2]{\langle{#1},{#2}\rangle} 
\renewcommand{\ip}[2]{\left\langle#1,#2\right\rangle}

\renewcommand{\R}{\mathbb{R}}

\newcommand{\distr}{\stackrel{d}{=}}
\newcommand{\A}{\mathcal{A}}

\newcommand{\M}{\mathcal{M}}
\newcommand{\cN}{{\bf \mathcal{N}}}

\newcommand{\ignore}[1]{\relax}

\newlength\myindent
\newtheorem{theorem}{Theorem}[section]

\newtheorem{lemma}[theorem]{Lemma}
\newtheorem{conjecture}[theorem]{Conjecture}

\newtheorem{proposition}[theorem]{Proposition}

\newtheorem{definition}[theorem]{Definition}

\renewcommand{\ip}[2]{\left\langle#1,#2\right\rangle}

\newcommand{\D}{\mathcal{D}}

\newcounter{parentnumber}

\makeatletter
\def\BState{\State\hskip-\ALG@thistlm}
\makeatother

\definecolor{Red}{rgb}{1,0,0}
\definecolor{Blue}{rgb}{0,0,1}
\definecolor{Olive}{rgb}{0.41,0.55,0.13}
\definecolor{Green}{rgb}{0,1,0}
\definecolor{MGreen}{rgb}{0,0.8,0}
\definecolor{DGreen}{rgb}{0,0.55,0}
\definecolor{Yellow}{rgb}{1,1,0}
\definecolor{Cyan}{rgb}{0,1,1}
\definecolor{Magenta}{rgb}{1,0,1}
\definecolor{Orange}{rgb}{1,.5,0}
\definecolor{Violet}{rgb}{.5,0,.5}
\definecolor{Purple}{rgb}{.75,0,.25}
\definecolor{Brown}{rgb}{.75,.5,.25}
\definecolor{Grey}{rgb}{.5,.5,.5}
\definecolor{Pink}{rgb}{1,0,1}
\definecolor{DBrown}{rgb}{.5,.34,.16}
\definecolor{Black}{rgb}{0,0,0}

\usepackage{float}

\usepackage{float}
\title{Geometric Barriers for Stable and Online Algorithms\\ for Discrepancy Minimization}
\author{{\sf David Gamarnik}\thanks{Sloan School of Management, Massachusetts Institute of Technology; e-mail: {\tt gamarnik@mit.edu}.} 
\and
{\sf Eren C. K{\i}z{\i}lda\u{g}}\thanks{Department of Statistics, Columbia University; e-mail: {\tt eck2170@columbia.edu}.}
\and
{\sf Will Perkins}\thanks{School of Computer Science, Georgia Institute of Technology; e-mail: {\tt math@willperkins.org}.}
\and 
{\sf Changji Xu}\thanks{Center of Mathematical Sciences and Applications, Harvard University; e-mail: {\tt cxu@cmsa.fas.harvard.edu}.}
}
\begin{document}
\maketitle
\begin{abstract}

For many computational problems involving randomness,  intricate geometric features of the solution space have been used to rigorously rule out powerful classes of algorithms. This is often accomplished through the lens of the multi Overlap Gap Property ($m$-OGP), a rigorous barrier against algorithms exhibiting input stability. In this paper, we focus on the algorithmic tractability of two models: (i) discrepancy minimization, and (ii) the symmetric binary perceptron (\texttt{SBP}), a random constraint satisfaction problem as well as a toy model of a single-layer neural network.

Our first focus is on the limits of online algorithms. By establishing and leveraging a novel geometrical barrier, we obtain sharp hardness guarantees against online algorithms for both the \texttt{SBP} and  discrepancy minimization. Our results match the best known  algorithmic guarantees, up to constant factors. Our second focus is on efficiently finding a constant discrepancy solution, given a random matrix $\M\in\R^{M\times n}$. In a smooth setting, where the entries of $\M$ are i.i.d.\,standard normal, we establish the presence of $m$-OGP for $n=\Theta(M\log M)$. Consequently, we rule out the class of stable algorithms at this value. These results give the first rigorous evidence towards a conjecture of Altschuler and Niles-Weed~\cite[Conjecture~1]{altschuler2021discrepancy}. 

Our methods use the intricate geometry of the solution space to prove tight hardness results for online algorithms. The barrier we establish is a novel variant of the $m$-OGP.  Furthermore, it regards $m$-tuples of solutions with respect to correlated instances, with growing values of $m$, $m=\omega(1)$.  Importantly, our results rule out online algorithms succeeding even with an exponentially small probability.
\end{abstract}
\newpage
\tableofcontents
\newpage

\section{Introduction}
In this paper, we study the \emph{discrepancy minimization} problem and the \emph{perceptron} model. Combinatorial discrepancy theory~\cite{spencer1985six,matousek1999geometric} is a central topic at the intersection of combinatorics, probability, and algorithms.  Given a matrix $\M\in\R^{M\times n}$, the central task in discrepancy theory is computing or bounding the quantity
\[
\mathcal{D}(\M) \triangleq \min_{\bs\in\Sigma_n}\bigl\|\M\bs\bigr\|_\infty,
\]
known as the \emph{discrepancy} of $\M$.

The \textit{perceptron} is a toy one-layer neural network model storing random patterns as well as a very natural high-dimensional probabilistic model, see~\cite{joseph1960number,winder1961single,wendel1962problem,cover1965geometrical} for early works on it. Given random patterns $X_i\in\R^n$, $1\le i\le M$,  \emph{storage} is achieved if one finds a $\bs \in\R^n$ `consistent' with all $X_i$: $\ip{\bs}{X_i}\ge 0$ for $1\le i\le M$. The vector $\bs$ is interpreted as \emph{synaptic weights}; it can either lie on the sphere in $\R^n$, $\|\bs\|_2 = \sqrt{n}$, or have binary entries, $\bs\in\Sigma_n=\{-1,1\}^n$. The former is dubbed as the spherical perceptron, see~\cite{gardner1988space,shcherbina2003rigorous,stojnic2013another,talagrand2011mean,alaoui2020algorithmic} for relevant work. In this paper, we only focus on the latter, dubbed as the \emph{binary perceptron}. A fundamental object studied in the  perceptron literature is the \emph{storage capacity}: the maximum number of (random) patterns that can be stored with a suitable $\bs$, see~\cite{gardner1987maximum,gardner1988space,gardner1988optimal}. Krauth and M{\'e}zard~\cite{krauth1989storage} gave a detailed though non-rigorous characterization of the storage capacity.  More recently, perceptron models with an \emph{activation function} $U:\R\to\{0,1\}$ are considered, where a pattern $X_i$ is stored with respect to (w.r.t.) $U$ if $U(\ip{\bs}{X_i})=1$. Of particular interest to us is the activation $U(x)=\mathbf 1_{|x| \le \kappa \sqrt{n}}$ which defines the \emph{symmetric binary perceptron} (\texttt{SBP}) model proposed by Aubin, Perkins, and Zdeborov{\'a}~\cite{aubin2019storage} (see also~\cite{bolthausen2021gardner,nakajima2023sharp} for results on more general perceptron models). As we see below, the \texttt{SBP} is closely related to discrepancy minimization.

\subsection{Discrepancy Minimization}
The discrepancy literature pertains to both \emph{worst-case} and \emph{average-case} $\M$. In the worst-case, minimal structure is assumed on $\M$, whereas in the average-case, the entries of $\M$ are random, e.g.\,i.i.d.\,Bernoulli, Rademacher, or standard normal. Moreover, both existential as well as algorithmic results are sought in discrepancy theory.

Concerning the worst-case analysis, a landmark result in the area is due to Spencer~\cite{spencer1985six}: $\mathcal{D}(\M)\le 6\sqrt{n}$ if $\M\in\R^{n\times n}$ with $|M_{ij}|\le 1$ for $1\le i,j\le n$ (`six standard deviations suffice'). The significance of this result is the improvement over the discrepancy guaranteed by the basic probabilistic method:  the discrepancy incurred by a random signing is of order $\Theta(\sqrt{n\log n})$ which is substantially larger than $O(\sqrt{n})$.  
It is worth noting that Spencer's result is worst-case and non-constructive, but recent works~\cite{bansal2010constructive,lovett2015constructive,levy2017deterministic,rothvoss2017constructive} has given efficient algorithms to find such low discrepancy solutions. 

In this paper we focus on average-case discrepancy. Suppose $\M\in\R^{M\times n}$ has i.i.d. $\cN(0,1)$ entries and $M=o(n)$. In this case, a line of work initiated in~\cite{karmarkar1986probabilistic} (for $M=1$) and subsequently continued in~\cite{costello2009balancing,turner2020balancing} (for $M\ge 2$) established that $\mathcal{D}(\M) = \Theta(\sqrt{n}2^{-n/M})$ w.h.p. Algorithmic results in this regime are found in~\cite{karmarkar1982differencing,yakir1996differencing,turner2020balancing}.  In the special case of $M=1$, \cite{gamarnik2021algorithmic} gave  rigorous evidence that finding a $\bs$ with $\|\M\bs\|_\infty = 2^{-\omega(\sqrt{n\log n})}$ may be algorithmically intractable. On the other hand when $M=\Theta(n)$, then it turns out $\mathcal{D}(\M) = \Theta(\sqrt{n})$ and this case is closely related to the \texttt{SBP}, see Section~\ref{sec:connections} for more details.

Next suppose the entries of $\M$ are i.i.d.\,binary, e.g.\,Rademacher or Bernoulli$(p)$. In this case, while still $\D(\M)=\Theta(\sqrt{n})$ w.h.p.\,when $M=\Theta(n)$, but it turns out that \textit{constant} discrepancy, in fact $\D(\M)=1$, is possible when $n$ is much larger than $M$. The sharpest possible result to this end is due to Altschuler and Niles-Weed~\cite{altschuler2021discrepancy} who completely resolved the question of exactly when $\D(\M)\le 1$: $\D(\M)\le 1$ if $\M$ consists of Bernoulli$(p)$ entries with arbitrary $p$ and $n\ge CM\log M$, where $C$ is any arbitrary constant greater than $(2\log 2)^{-1}$. Their result covers in particular the sparse regime, $p=o(1)$, and is the sharpest possible as $n=\Omega(M\log M)$ is needed for $\D(\M)$ to be $O(1)$: if $n=CM\log M$ for $C<(2\log 2)^{-1}$ and $p=1/2$, then w.h.p. no constant discrepancy solutions exist. Earlier results towards this direction are found in~\cite{hoberg2019fourier,franks2020discrepancy,potukuchi2018discrepancy}. 
 Equipped with this existential guarantee from~\cite{altschuler2021discrepancy}, a natural algorithmic question is whether one can find such a constant discrepancy solution in polynomial time. This task is conjecturally hard, see~\cite[Conjecture~1]{altschuler2021discrepancy}. The algorithmic tractability of this problem is a main focuses of the present paper.
\subsection{Symmetric Binary Perceptron (\texttt{SBP})}
Fix  $\kappa>0$,  $\alpha>0$, and set $M=\lfloor n\alpha\rfloor\in\mathbb{N}$. Let $X_i\sim \cN(0,I_n)$, $1\le i\le M$, be i.i.d.\,random vectors, where $\cN(0,I_n)$ is the centered multivariate normal distribution in $\R^n$ with identity covariance. Consider the (random) set
\begin{equation}\label{eq:sbp-set}
    S_\alpha(\kappa) =\Bigl\{\bs\in\Sigma_n: \left|\ip{\bs}{X_i}\right|\le \kappa\sqrt{n},1\le i\le M \Bigr\} = \Bigl\{\bs\in\Sigma_n:\bigl\|\M\bs\bigr\|_\infty\le\kappa\sqrt{n}\Bigr\},
\end{equation}
where $\M\in\R^{M\times n}$ with rows $X_1,\dots,X_M$. The word \emph{symmetric} refers to the fact $\bs\in S_\alpha(\kappa)$ iff $-\bs\in S_\alpha(\kappa)$. The \texttt{SBP} was put forth by Aubin, Perkins, and Zdeborov{\'a}~\cite{aubin2019storage} as a symmetric counterpart to the \emph{asymmetric binary perceptron} (\texttt{ABP}), where the constraints are instead of form $\ip{\bs}{X}\ge \kappa\sqrt{n}$, $1\le i\le M$. The \texttt{ABP} turns out very challenging mathematically, see~\cite{krauth1989storage,kim1998covering,talagrand1999intersecting,xu2019sharp,ding2019capacity,perkins2021frozen,abbe2021binary,gamarnik2022algorithms,kizildag2022algorithms} for relevant work and more details. The \texttt{SBP}, on the other hand, retains pertinent structural properties conjectured for the \texttt{ABP}~\cite{baldassi2020clustering}, while being more amenable to rigorous analysis thanks to the symmetry.

The \texttt{SBP} undergoes a \emph{sharp phase transition}, conjectured in~\cite{aubin2019storage} and subsequently proven independently by Perkins and Xu~\cite{perkins2021frozen} and Abbe, Li, and Sly~\cite{abbe2021proof}. Let
\begin{equation}\label{eq:alpha-c-kappa}
    \alpha_c(\kappa) \triangleq -\frac{1}{\log_2\mathbb{P}[|Z|\le \kappa]},\quad \text{where}\quad Z\sim \cN(0,1).
\end{equation}
Then
\begin{equation}\label{eq:pt}
\lim_{n\to\infty}\mathbb{P}\bigl[S_\alpha(\kappa)\ne\varnothing\bigr] =\begin{cases}0,&\text{if } \alpha>\alpha_c(\kappa)\\1,&\text{if } \alpha<\alpha_c(\kappa)
\end{cases}.
\end{equation}
The part $\alpha>\alpha_c(\kappa)$ is due to~\cite{aubin2019storage} and established via an application of the \emph{first moment method}: $\mathbb{E}\bigl[|S_\alpha(\kappa)|\bigr]=o(1)$ if $\alpha>\alpha_c(\kappa)$, so $S_\alpha(\kappa)=\varnothing$ w.h.p.\,by Markov's inequality. The same paper also studies the case $\alpha<\alpha_c(\kappa)$ and establishes, through the \emph{second moment method}, that $\liminf_{n\to\infty}\mathbb{P}\bigl[S_\alpha(\kappa)\ne\varnothing\bigr]>0$. Boosting this to a high probability guarantee requires more powerful tools, see~\cite{perkins2021frozen} for a delicate martingale argument and~\cite{abbe2021proof} for an argument based on a fully connected analog of the small subgraph conditioning method. Furthermore, very recently, the critical window around $\alpha_c(\kappa)$ was shown to be of constant width, see~\cite{altschuler2022fluctuations,sah2023}.
These facts highlight that the first moment `prediction' for the precise location of the phase transition is correct, and the transition is very sharp. 

Given that $S_\alpha(\kappa)$ is w.h.p.\,non-empty if $\alpha<\alpha_c(\kappa)$, a natural follow-up question is algorithmic: can  a solution $\bs\in S_\alpha(\kappa)$ be found efficiently?  And does the existence of efficient algorithms depend on $\alpha$?  
Efficient algorithms at small densities $\alpha$ were given in~\cite{kim1998covering} and \cite{abbe2021binary} for the \texttt{ABP} and \texttt{SBP} respectively while on the negative side, \cite{gamarnik2022algorithms} studied the limits of efficient algorithms (see details below).  The works~\cite{baldassi2007efficient,baldassi2015subdominant,baldassi2020clustering} put forth possible explanations for the success of efficient algorithms: while almost all solutions are totally frozen (conjectured in~\cite{mezard2005clustering,huang2014origin}), efficient algorithms access rare solutions lying in large clusters.  Recent works including~\cite{perkins2021frozen,abbe2021proof,abbe2021binary} have studied these structural predictions.
\subsection{Connections between Discrepancy Theory and the \texttt{SBP}}\label{sec:connections}
In order to explicate the connection between discrepancy minimization and the \texttt{SBP}, we focus on the \emph{proportional regime}, i.e. $M=\Theta(n)$. The discrepancy viewpoint is to take an  $\M\in\R^{M\times n}$ with a fixed aspect ratio $\alpha=M/n$, and to seek a $\bs$ such that $\|\M\bs\|_\infty$ is as small as possible. The perceptron viewpoint, on the other hand, is the inverse: fix a $\kappa>0$ first and seek the largest $\alpha$ for which a solution $\bs$ with $\|\M\bs\|_\infty\le \kappa\sqrt{n}$ exists. Furthermore, the asymptotic value of the average-case discrepancy in the proportional regime immediately follows from the sharp threshold result for the \texttt{SBP}~\eqref{eq:pt}: $\mathcal{D}(\M) = (1+o(1))f(\alpha)\sqrt{n}$ w.h.p., where $f(\alpha)$ is the `inverse' of $\alpha_c(\kappa)$~\eqref{eq:alpha-c-kappa}. 
\paragraph{Algorithmic Connections} 
The connection between the \texttt{SBP} and discrepancy theory further extends to algorithmic domain: the best known efficient algorithm for the \texttt{SBP} comes from the discrepancy literature. Suppose $\M\in\R^{M\times n}$ has i.i.d.\,Rademacher entries. Bansal and Spencer~\cite{bansalspenceronline} devised an efficient \emph{online algorithm} that finds a $\bs_{\rm ALG}\in\Sigma_n$ such that $\|\M\bs_{\rm ALG}\|_\infty =O(\sqrt{M})$ w.h.p.\,if $n\ge M=\omega(1)$. Informally, an algorithm is online if the $t^{\rm th}$ coordinate of the output $\bs_{\rm ALG}$ depends only on first $t$ columns of $\M$, see Definition~\ref{def:online-algs} for a formal definition. As an immediate corollary, this yields an efficient algorithm for the \texttt{SBP} that finds a solution $\bs\in S_\alpha(\kappa)$ w.h.p.\,if $\alpha = O(\kappa^2)$, see~\cite[Corollary~3.6]{gamarnik2022algorithms}. In fact, this is the best known algorithmic guarantee both for the \texttt{SBP} and for  discrepancy in the random proportional regime, see~\cite[Section~3.3]{gamarnik2022algorithms}.

In light of these existential and algorithmic results, it appears that the \texttt{SBP} may exhibit a striking \emph{statistical-to-computational gap} (\texttt{SCG}): the density below which solutions exist w.h.p., i.e.\,$\alpha_c(\kappa)$, is substantially larger than those below which polynomial-time search algorithms work. Further, this \texttt{SCG} is most profound when $\kappa\to 0$. While the Bansal-Spencer algorithm works only when $\alpha=O(\kappa^2)$, solutions do exist w.h.p.\,below $\alpha_c(\kappa)$ which, per~\eqref{eq:alpha-c-kappa}, is asymptotically $\frac{1}{\log_2(1/\kappa)}$. 
Origins of this \texttt{SCG} were investigated in~\cite{gamarnik2022algorithms}, where it was shown that the \texttt{SBP} exhibits an intricate geometrical property called the \emph{multi Overlap Gap Property} ($m$-OGP) when $\alpha=\Omega(\kappa^2 \log_2\frac1\kappa)$ and consequently \emph{stable algorithms} fail to find a satisfying solution for $\alpha = \Omega(\kappa^2 \log_2 \frac1\kappa)$. It is worth noting, though, that stable algorithms need not include online algorithms, which achieve the computational threshold for the \texttt{SBP}. What the limits of online algorithms are is an open question we undertake in this paper.

In addition to the \texttt{SBP}, the discrepancy minimization problem --- in particular the algorithmic problem of efficiently finding a constant discrepancy solution when such solutions exist w.h.p. --- also exhibits a  similar \texttt{SCG}. To recall, when $\M\in\R^{M\times n}$, then constant discrepancy solutions exist w.h.p.\,as soon as $n=\Omega(M\log M)$. On the other hand, the best known polynomial-time algorithm succeeds at a dramatically smaller value $M=o(\log n)$~\cite{bansalpersonal}, highlighting another striking \texttt{SCG}. This is our second focus in the present paper.

\subsection{Main Results}

Suppose $\M$ per~\eqref{eq:sbp-set} consists of i.i.d.\,$\cN(0,1)$ entries. Our first main result establishes that online algorithms fail to find a satisfying solution for the \texttt{SBP} at densities $\alpha=\Omega(\kappa^2)$.  
\begin{theorem}[Informal, see Theorem~\ref{thm:online-sbp}]\label{thm:online-sbp-informal}  For densities $\alpha=\Omega(\kappa^2)$, online algorithms fail to find a solution for the \texttt{SBP} w.p.\,greater than $e^{-\Theta(n)}$. 
\end{theorem}
Our next result extends Theorem~\ref{thm:online-sbp-informal} to the discrepancy minimization problem when $\M\in\R^{M\times n}$ consists of i.i.d.\,Rademacher or i.i.d.\,Bernoulli$(p)$ entries. 
\begin{theorem}[Informal, see Theorems~\ref{thm:online-disc}-\ref{thm:online-disc-bern}]\label{thm:online-disc-informal} 
There exists $c>0$ such that online algorithms fail to return a solution of discrepancy at most $c\sqrt{M}$ w.p.\,greater than $e^{-\Theta(M)}$.
\end{theorem}
If the entries of $\M$ are Rademacher, taking $c=1/24$ suffices.
For Bernoulli case, the implied constant depends on $p$: it suffices to take $c\triangleq c_p = \sqrt{p-p^2}/24$. Taken together, Theorems~\ref{thm:online-sbp-informal} and~\ref{thm:online-disc-informal} collectively yield that among the class of online algorithms, Bansal-Spencer algorithm~\cite{bansalspenceronline} is optimal up to constants for both models. Our proof is based on a novel version of $m$-OGP: we show the non-existence of tuples of solutions agreeing on first $1-\Delta$ fraction of coordinates for a suitable $\Delta\in(0,1)$, for a collection of $m$ correlated instances, see below for details. This barrier is more restricted than $m$-OGP, which asserts the non-existence of tuples of solutions at a prescribed distance. Additionally, for Theorem~\ref{thm:online-disc-informal}, one has to consider $m$-tuples with growing values of $m$, $m=\omega(1)$; this idea is originally due to Gamarnik and K{\i}z{\i}lda\u{g}~\cite{gamarnik2021algorithmic} for lowering the $m$-OGP threshold.

To the best of our knowledge, Theorems~\ref{thm:online-sbp-informal}-\ref{thm:online-disc-informal} are the first (up-to-constants) tight hardness guarantees via geometrical barriers against classes beyond stable algorithms, 
see Section~\ref{sec:background} for details. Furthermore, unlike prior work~\cite{gamarnik2020low,wein2020optimal,huang2021tight,gamarnik2021algorithmic,gamarnik2022algorithms}, the algorithms ruled out need not succeed w.h.p.\,or even with a constant probability: an exponentially small success probability suffices. This is made possible by using a clever application of Jensen's inequality, originally due to Gamarnik and Sudan~\cite{gamarnik2017performance}.

\paragraph{Proof Sketch} We sketch the proof of Theorem~\ref{thm:online-sbp-informal}, which is based on a new version of $m$-OGP coupled with a contradiction argument. Suppose that such an online algorithm $\A$ with a success probability of $p_s$ exists. Let $\M_1\in\R^{M\times n}$ with i.i.d.\,$\cN(0,1)$ entries. Fix an $m\in\mathbb{N}$ and a $\Delta\in(0,1)$, generate random matrices $\M_i\in\R^{M\times n}$, $2\le i\le m$, by independently resampling the last $\Delta n$ columns of $\M_1$. Running $\A$ on each $\M_i$, we obtain solutions $\bs_i \triangleq \A(\M_i)\in\Sigma_n$, $1\le i\le m$. An application of Jensen's inequality then reveals $\|\M_i\bs_i\|_\infty\le \kappa\sqrt{n}$, $1\le i\le m$, w.p.\,at least $p_s^m$. Furthermore, since $\A$ is online, it is the case that any $\bs_i$ and $\bs_j$ necessarily have identical first $n-\Delta n$ coordinates. Namely, if such an $\A$ exists, then w.p.\,at least $p_s^m$, there exists an $m$-tuple $(\bs_1,\dots,\bs_m)$ of satisfying solutions that agree on first $n-\Delta n$ coordinates. We then establish, using the \emph{first moment method}, that for suitably chosen $m,\Delta$; the probability that such an $m$-tuple exists is in fact strictly less than $p_s^m$. This is a contradiction. The proof of Theorem~\ref{thm:online-disc-informal} is similar, though it requires additional technical steps. In particular, one needs an anti-concentration argument for signed sums of binary variables via Berry-Esseen Theorem.


Our next focus is on the algorithmic problem of efficiently finding a constant discrepancy solution, given a random $\M\in\R^{M\times n}$. To recall, such solutions exist w.h.p.\,as soon as $n=\Omega(M\log M)$~\cite{altschuler2021discrepancy}, while the best known polynomial-time algorithm works only when $M=o(\log n)$~\cite{bansalpersonal}. Further, it was conjectured in~\cite{altschuler2021discrepancy} that this task is algorithmically hard. Towards this conjecture, we focus on a smooth setting where the entries of $\M$ are i.i.d.\,$\cN(0,1)$. Our next main result shows the presence of $m$-OGP with $m=O(1)$ when $n=\Theta(M\log M)$, giving a rigorous evidence of hardness at the `boundary' $n=\Theta(M\log M)$.
\begin{theorem}[Informal, see Theorem~\ref{thm:m-ogp-discrepancy}]\label{thm:m-ogp-disc-informal} 
    For $n=\Theta(M\log M)$, the set of constant discrepancy solutions exhibits $m$-OGP (with constant $m$) for suitably chosen parameters.
\end{theorem}
The regime $\log n\ll M \ll n/\log n$ as well as extensions beyond Gaussian disorder---in particular to the Bernoulli or Rademacher case---are among the open problems we discuss in Section~\ref{sec:open-prob}. 

Our final main result leverages the $m$-OGP to show that \emph{stable algorithms} fail to find a constant discrepancy solution when $n=\Theta(M\log M)$. Informally, an algorithm is stable if a small perturbation of its inputs induces only a small change in its output $\bs$, see Definition~\ref{def:stable-alg} for a formal statement. The class of stable algorithm has been shown to capture powerful classes of algorithms including low-degree polynomials~\cite{gamarnik2020low,bresler2021algorithmic}, Approximate Message Passing (AMP)~\cite{gamarnikjagannath2021overlap}, and Boolean circuits of low-depth~\cite{gamarnik2021circuit}.
\begin{theorem}[Informal, see Theorem~\ref{thm:stable-hard}]\label{thm:alg-hard-disc-informal}
    For $n=\Theta(M\log M)$, stable algorithms fail to find a constant discrepancy solution w.p.\,greater than a certain constant.
\end{theorem}
The proof of Theorem~\ref{thm:alg-hard-disc-informal} is based on a Ramsey-theoretic argument developed in~\cite{gamarnik2021algorithmic} and also used in~\cite{gamarnik2022algorithms} coupled with the $m$-OGP result, Theorem~\ref{thm:m-ogp-discrepancy}; it rules out stable algorithms succeeding with a constant probability.
\subsection{Background and Prior Work}\label{sec:background}
\paragraph{Statistical-to-Computational Gaps (\texttt{SCG}s)} Both the \texttt{SBP} and discrepancy minimization exhibit an \texttt{SCG}: known efficient algorithms perform strictly worse than the existential guarantee. Such gaps are a universal feature of many \emph{average-case} algorithmic problems arising from random combinatorial structures and high-dimensional statistical inference. A partial list of problems with an \texttt{SCG} include random CSPs~\cite{mezard2005clustering,achlioptas2006solution,achlioptas2008algorithmic,gamarnik2017performance,bresler2021algorithmic}, optimization over random graphs~\cite{gamarnik2014limits,coja2015independent,wein2020optimal}, spin glasses~\cite{gamarnikjagannath2021overlap,huang2021tight}, planted clique~\cite{deshpande2015improved,barak2019nearly}, and tensor decomposition~\cite{wein2022average}, see also the surveys by Gamarnik~\cite{gamarnik2021overlap} and Gamarnik, Moore, and Zdeborov{\'a}~\cite{gamarnik2022disordered}. Unfortunately, standard computational complexity theory is often useless due to the average-case nature of such problems\footnote{Modulo a few exceptions, see e.g.~\cite{ajtai1996generating,boix2021average, GK-SK-AAP}.}. Nevertheless, a very promising line of research proposed various frameworks that provide \emph{rigorous evidence} of hardness. These frameworks include average-case reductions---often from the planted clique~\cite{berthet2013computational,brennan2018reducibility,brennan2019optimal}---as well as unconditional lower bounds against restricted classes of algorithms, including the statistical query algorithms~\cite{diakonikolas2017statistical,feldman2017statistical,feldman2018complexity}, low-degree polynomials~\cite{hopkins2018statistical,kunisky2022notes,wein2022average}, sum-of-squares hierarchy~\cite{hopkins2015tensor,hopkins2017power,raghavendra2018high,barak2019nearly}, AMP~\cite{zdeborova2016statistical,bandeira2018notes}, and Monte Carlo Markov Chain (MCMC) methods~\cite{jerrum1992large,dyer2002counting}. Yet another such approach is based on the intricate geometry of the solution space through the \emph{Overlap Gap Property}.
\paragraph{Intricate Geometry and the Overlap Gap Property (OGP)} Prior work~\cite{mezard2005clustering,achlioptas2006solution,achlioptas2008algorithmic} discovered a very intriguing connection between intricate geometry and algorithmic hardness in the context of random CSPs: the onset of algorithmic hardness roughly coincides with the emergence of an intricate geometry in the solution space. The OGP framework leverages insights from statistical physics to rigorously link the intricate geometry to formal hardness. In the context of random optimization, the OGP informally states that (w.h.p.\,over the randomness) any two near-optima are either `close' or `far': there exists $0<\nu_1<\nu_2<1$ such that $n^{-1}\ip{\bs}{\bs'}\in[0,\nu_1]\cup[\nu_2,1]$ for any pair of near-optima $\bs,\bs'\in\Sigma_n$. Namely, the region of normalized overlaps is \emph{topologically disconnected}; no pairs of near-optima at \emph{intermediate} distances can be found. The OGP is a rigorous barrier for large classes of algorithms exhibiting input stability---see below. See~\cite{gamarnik2021overlap} for a survey on OGP. 

\paragraph{Algorithmic Implications of OGP} The first work establishing and leveraging OGP to rule out algorithms is due to Gamarnik and Sudan~\cite{gamarnik2014limits,gamarnik2017}. Their focus is on the problem finding a large independent set in sparse random graphs on $n$ vertices with average degree $d$, which exhibits an $\texttt{SCG}$: the largest  such independent set is of size $2\frac{\log d}{d}n$~\cite{frieze1992independence,bayati2010combinatorial}, whereas the best known efficient algorithm finds an independent set of size $\frac{\log d}{d}n$. They establish that any pair of independent sets of size larger than $(1+1/\sqrt{2})\frac{\log d}{d}n$ exhibit the OGP. By leveraging the OGP, they then show that \emph{local algorithms} fail to find an independent set of size larger than  $(1+1/\sqrt{2})\frac{\log d}{d}n$. Subsequent research established and leveraged the OGP to rule out other classes of algorithms (e.g., AMP~\cite{gamarnikjagannath2021overlap}, low-degree polynomials~\cite{gamarnik2020low,wein2020optimal,bresler2021algorithmic}, Langevin dynamics~\cite{gamarnik2020low,huang2021tight}, low-depth circuits~\cite{gamarnik2021circuit}) for various other models (e.g., random graphs~\cite{rahman2017local,gamarnik2020low,wein2020optimal}, spin glass models~\cite{chen2019suboptimality,gamarnikjagannath2021overlap,huang2021tight}, random CSPs~\cite{gamarnik2017performance,bresler2021algorithmic,gamarnik2022algorithms}). A very important feature found across the algorithms ruled out by the OGP and other versions of intricate geometry is input stability, similar to Definition~\ref{def:stable-alg} (apart from the failure of MCMC in \emph{planted} models, see e.g.~\cite{jerrum1992large,arous2020free,gamarnik2021overlapsubmatrix}).  Our work marks the first instance of intricate geometry yielding tight  algorithmic hardness against classes beyond stable algorithms.

\paragraph{Multi OGP ($m$-OGP)} The prior work~\cite{gamarnik2014limits,gamarnik2017} discussed above establish the failure of local algorithms at value $(1+1/\sqrt{2})\frac{\log d}{d}n$. By considering a certain overlap pattern involving many large independent sets, Rahman and Vir{\'a}g~\cite{rahman2017local} subsequently removed the additional $1/\sqrt{2}$ term; they showed that the onset of  OGP precisely coincides with the algorithmic $\frac{\log d}{d}n$ value.
That is, one can potentially lower the onset of the OGP and rule out algorithms for a broader range of parameters through more intricate overlap patterns. In a similar vein, Gamarnik and Sudan~\cite{gamarnik2017performance} studied the Not-All-Equal $k$-SAT problem and showed the presence of the OGP for the $m$-tuples of nearly equidistant satisfying assignments. Consequently, they obtained nearly tight hardness guarantees against sequential local algorithms. A similar $m$-OGP was also employed in~\cite{gamarnik2021algorithmic,gamarnik2022algorithms}, and is also our focus here. 

Recently, $m$-OGP for more intricate patterns were proposed. These forbidden patterns regard $m$-tuples of solutions where for any $2\le i\le m$, the $i^{\rm th}$ solution has `intermediate' overlap with the first $i-1$ solutions. By doing so, tight hardness guarantees against low-degree polynomials were obtained for finding independent sets in sparse random graphs by Wein~\cite{wein2020optimal} and for the random $k$-SAT by Bresler and Huang~\cite{bresler2021algorithmic}. Similarly, Huang and Sellke~\cite{huang2021tight} construct a very intricate forbidden structure consisting of an ultrametric tree of solutions dubbed as the \emph{Branching OGP}. By leveraging the branching OGP, they obtain tight hardness guarantees against Lipschitz algorithms for the $p$-spin model. Moreover, these papers establish the \emph{Ensemble $m$-OGP} which regards $m$-tuples that are near-optimal w.r.t.\,correlated instances. The Ensemble OGP emerged in~\cite{chen2019suboptimality}; it has been instrumental in ruling out stable algorithms since. The investigation of $m$-tuples of solutions w.r.t.\,correlated instances is at the core of our paper. Furthermore, we inspect $m$-tuples with super-constant $m$, $m=\omega(1)$, to rule out online algorithms in Theorems~\ref{thm:online-disc}-\ref{thm:online-disc-bern}. This idea originated in~\cite{gamarnik2021algorithmic} for further lowering the $m$-OGP threshold. 
\subsection{Open Problems}\label{sec:open-prob}
\paragraph{Geometrical Barriers for other Classes of Algorithms} Prior work on OGP showed that intricate geometry is a signature of algorithmic hardness, and gave lower bounds against stable algorithms. Theorems~\ref{thm:online-sbp},~\ref{thm:online-disc} and~\ref{thm:online-disc-bern} extend this beyond stable algorithms; they leverage intricate geometry to rule out online algorithms. It would be very interesting to rule out other classes of algorithms via  similar geometrical barriers.
\paragraph{Discrepancy Minimization beyond Gaussian Disorder} Theorem~\ref{thm:m-ogp-discrepancy} shows that, for $\M\in\R^{M\times n}$ with i.i.d. $\cN(0,1)$ entries and $n=\Theta(M\log M)$, the set of constant discrepancy solutions exhibits the $m$-OGP. A very interesting question is whether $m$-OGP still holds when the entries of $\M$ are binary. Prior work established OGP both for models with discrete disorder (e.g., random $k$-SAT~\cite{gamarnik2017performance,bresler2021algorithmic}, random graphs~\cite{gamarnik2014limits,gamarnik2017,gamarnik2020low,wein2020optimal}) as well as for models with continuous disorder (e.g., spin glass models~\cite{gamarnikjagannath2021overlap,huang2021tight}, number partitioning~\cite{gamarnik2021algorithmic}, the \texttt{SBP}~\cite{gamarnik2022algorithms}). These results suggest that OGP exhibits universality: the distributional details of the disorder are immaterial\footnote{In fact, such a universality result has already been established for the \texttt{SBP}, see~\cite[Theorem~5.2]{gamarnik2022algorithms}.}. In light of these, we make the following conjecture:
\begin{conjecture}\label{conj:ogp-universality}
For $\M\in\R^{M\times n}$ with i.i.d.\,Rademacher or Bernoulli$(p)$ entries and $n=\Theta(M\log M)$, set of constant discrepancy solutions exhibits $m$-OGP with suitable parameters.
\end{conjecture}
Resolving Conjecture~\ref{conj:ogp-universality} may require understanding a probability term of the form $\mathbb{P}[M\boldsymbol{v}=\boldsymbol{x}]$ for a random $\boldsymbol{v}$ with i.i.d.\,binary entries and a deterministic $M\in\{-1,1\}^{m\times n}$ whose rows $\bs_1,\dots,\bs_m$ satisfy $n^{-1}\ip{\bs_i}{\bs_j}=\beta$ for some fixed $\beta$ and every $1\le i<j\le m$. 
One direction is to employ local limit arguments; we leave this as an open problem for future investigation.
\paragraph{Discrepancy Minimization beyond $n=\Theta(M\log M)$} Recall that constant discrepancy solutions exist as soon as $n=\Theta(M\log M)$, i.e.\,when $M=O(n/\log n)$, while the best known polynomial-time algorithm works only when $M=o(\log n)$. In light of these, Theorem~\ref{thm:m-ogp-discrepancy} provides rigorous evidence of hardness, yet only at the `boundary'. The regime $\log n\ll M\ll n/\log n$ is an interesting direction left for future work. A potential avenue would be to consider a more intricate overlap pattern, such as those in~\cite{wein2020optimal,huang2021tight} or the \emph{branching OGP}~\cite{huang2021tight}. To this end, we discover an intriguing phase transition (proof omitted):
\begin{theorem}\label{thm:phase-tran}
    Let $\M\in\R^{M\times n}$ with i.i.d. $\cN(0,1)$ entries. Fix a $K>0$ and let $S(m,\delta,K)$ be the set of $(\bs_1,\dots,\bs_m)$ such that $d_H(\bs_i,\bs_j) = \delta$ and $\max_{i\le m}\|\M\bs_i\|_\infty \le K$.
    \begin{itemize}
        \item[(a)] If $M=\omega(\log n)$ then $\mathbb{E}\bigl[\bigl|S\bigl(n\log^{-O(1)} n,\log^{O(1)} n,K\bigr)\bigr|\bigr]=o(1)$.
        \item[(b)] If $M=o(\log n)$, then $\mathbb{E}\bigl[\bigl|S\bigl(n\log^{-O(1)} n,\log^{O(1)} n,K\bigr)\bigr|\bigr]=\omega(1)$.
    \end{itemize}
\end{theorem}
Namely, the value $M=\log n$ is the threshold at which the (expected) number of $m$-tuples of constant discrepancy solutions at distance $\log^{O(1)}n$ with $m=\widetilde{\Theta}(n)$ undergoes a phase transition. Whether this phase transition at value $\log n$ is coincidental or a signature of algorithmic hardness is an open problem left for future work. 
\paragraph{Paper Organization and Notation} The rest of the paper is organized as follows. In Section~\ref{sec:online-hard}, we formalize the class of online algorithms and state our hardness results. Section~\ref{sec:const-disc} is devoted to the algorithmic problem of finding constant discrepancy solutions. See Section~\ref{sec:prelim} for preliminaries and the definition of set of $m$-tuples regarding OGP, Section~\ref{sec:ogp} for the main $m$-OGP result and Section~\ref{sec:stable-hard} for the hardness result against stable algorithms. Finally, see Section~\ref{sec:pfs} for complete proofs. Our notation is fairly standard, see the beginning of Section~\ref{sec:pfs} for details.
\section{Tight Hardness Guarantees for Online Algorithms}\label{sec:online-hard} 
In this section, we explore the limits of \emph{online algorithms} in the context of 
our two-models: the \texttt{SBP} and  average-case discrepancy. We begin by formalizing the class of online algorithms in the context of these two models.
\begin{definition}\label{def:online-algs}
   \phantom{//} 
  \begin{itemize}
        \item {\bf (\texttt{SBP})} Fix a $\kappa>0$ and an $\alpha<\alpha_c(\kappa)$. Let $\M\in\R^{\alpha n\times n}$ with i.i.d.\,$\cN(0,1)$ entries.
        \item {\bf (Discrepancy)} Fix $n\ge M$, let $\M\in\R^{M\times n}$ has i.i.d.\,Rademacher or Bernoulli$(p)$ entries.
    \end{itemize}
    Fix a $p_s>0$ and a $K>0$. An algorithm $\A$ is  $(p_s,\alpha)$-online for the \texttt{SBP} or $(p_s,K)$ online for discrepancy if it satisfies the following. Let $\A(\M)\triangleq \bs_{\rm ALG}=(\bs_{\rm ALG}(i):1\le i\le n)\in\Sigma_n$. 
    \begin{itemize}
       \item {\bf (Success)} We have 
        \begin{align*}
            &\mathbb{P}\bigl[\bigl\|\M\bs_{\rm ALG}\bigr\|_\infty \le \kappa\sqrt{n}\bigr] \ge p_s \quad \text{(for the \texttt{SBP})} \\
            &\mathbb{P}\bigl[\bigl\|\M\bs_{\rm ALG}\bigr\|_\infty \le K\bigr] \ge p_s \quad \text{(for discrepancy)}.
        \end{align*}
        \item {\bf (Online)} Let $\mathcal{C}_1,\dots,\mathcal{C}_n$ be the columns of $\M$. There exists deterministic functions $f_1,\dots,f_n$ such that for $1\le t\le n$, $\bs_{\rm ALG}(t) = f_t\bigl(\mathcal{C}_1,\dots,\mathcal{C}_t\bigr)\in\{-1,1\}$.
    \end{itemize}
\end{definition}
The parameter $p_s$ is the success guarantee of the algorithm, where the probability is taken w.r.t.\,the randomness in $\M$. The online nature of the algorithm admits the following interpretation. Columns $\mathcal{C}_i$ arrive at a time. At the end of round $t-1$, the signs $\bs(i)\in\{-1,1\}$, $1\le i\le t-1$ are assigned, and a new column $\mathcal{C}_t$ arrives. The sign $\bs(t)$ then depends only on the previous decisions $\bs(i)$, $1\le i\le t-1$ and $\mathcal{C}_t$. That is, $\bs(t)$ depends only on $\mathcal{C}_i$, $1\le i\le t$. This abstraction captures, in particular, the Bansal-Spencer algorithm:
\begin{theorem}{\cite[Theorem~3.4]{bansalspenceronline}}
\label{thm:bansal-spencer}
    Let $n\ge M$ and $\M\in\{-1,1\}^{M\times n}$ has i.i.d.\,Rademacher entries. Then, there exists absolute constants $C>0$ and $\gamma<1$, and an online algorithm $\A$ admitting $\M$ as its input and returning a $\bs\triangleq \A(\M)$ such that
    \[    \mathbb{P}\bigl[\bigl\|\M\bs\bigr\|_\infty\le C\sqrt{M}\bigr]\le 1-e^{-\Theta(M^\gamma)}.
    \]
\end{theorem}
Theorem~\ref{thm:bansal-spencer} immediately yields an efficient algorithm for the \texttt{SBP}  when $\alpha\le \kappa^2/C^2$, see~\cite[Corollary~4.6]{gamarnik2022algorithms}. As mentioned in the introduction, the Bansal-Spencer algorithm is the best known polynomial-time algorithm both for the \texttt{SBP} and for the discrepancy minimization in random proportional regime. In the sense of Definition~\ref{def:online-algs}, it is a $(1-e^{-\Theta(n^\gamma)},\kappa^2/C^2)$-online algorithm for the \texttt{SBP} and a $(1-e^{-\Theta(M^\gamma)},C\sqrt{M})$-online algorithm for the discrepancy. 
\paragraph{Online Algorithms for the \texttt{SBP}} Our first main result focuses on the \texttt{SBP} in the regime $\kappa\to 0$ and establishes the following hardness for the class of online algorithms. 
\begin{theorem}\label{thm:online-sbp}
Fix any small enough $\kappa>0$ and any $\alpha\ge 4\kappa^2$. Then there exists an $n_0\in\mathbb{N}$ and an absolute constant $c>0$ such that the following holds. For any $n\ge n_0$, there exists no $(e^{-cn},\alpha)$-online algorithm for the \texttt{SBP} in the sense of Definition~\ref{def:online-algs}. 
\end{theorem}
We prove Theorem~\ref{thm:online-sbp} in Section~\ref{sec:pf-online-sbp}. Several remarks are in order.

Theorem~\ref{thm:online-sbp} establishes that in the regime $\kappa\to 0$, online algorithms fail to find a satisfying solution for the \texttt{SBP} for densities $\alpha=\Omega(\kappa^2)$. This substantially improves upon an earlier result in~\cite[Theorem~7.4]{gamarnik2022algorithms}, which showed the failure of online algorithms only when $\alpha$ is sufficiently close to the satisfiability threshold $\alpha_c(\kappa)$. Further, in light of Theorem~\ref{thm:bansal-spencer}, Theorem~\ref{thm:online-sbp} is the sharpest possible: Bansal-Spencer algorithm is optimal (up to constants) among online algorithms; no online algorithm, in the sense of Definition~\ref{def:online-algs}, can improve upon it. 

The algorithms Theorem~\ref{thm:online-sbp} rules out need not succeed w.h.p.\,or even with a constant probability: an exponentially small success guarantee suffices. This is a particular strength of Theorem~\ref{thm:online-sbp}; we are unaware of any similar hardness guarantees for algorithms that succeed w.p.\,$o(1)$. This is based on a clever application of Jensen's inequality that is originally due to Gamarnik and Sudan~\cite{gamarnik2017performance}.

\paragraph{Online Algorithms for the Discrepancy Minimization} Our second main result extends Theorem~\ref{thm:online-sbp} to  discrepancy minimization for the case when the entries of $\M$ are binary.
\begin{theorem}\label{thm:online-disc}
    Let $c<1/2$ be arbitrary, $n\ge M=\omega(1)$, and $\M\in\{-1,1\}^{M\times n}$ with i.i.d.\,Rademacher entries. Then there exists an $M_0\in\N$ such that the following holds. For every $M\ge M_0$, there exists no $\bigl(e^{-cM},\sqrt{M}/24\bigr)$-online algorithm for discrepancy in the sense of Definition~\ref{def:online-algs}.
\end{theorem}
Furthermore, Theorem~\ref{thm:online-disc} remains valid even when the entries of $\M$ are Bernoulli$(p)$.
\begin{theorem}\label{thm:online-disc-bern}
    Let $c<1/2$ be arbitrary, $n\ge M=\omega(1)$, and $\M\in\{0,1\}^{M\times n}$ with i.i.d.\,Bernoulli$(p)$ entries. Then there exists an $M_0\in\N$ such that the following holds. For every $(p-p^2)M\ge M_0$, there exists no $\bigl(e^{-cM},\sqrt{M(p-p^2)}/24\bigr)$-online algorithm for discrepancy in the sense of Definition~\ref{def:online-algs}.
\end{theorem} 
We prove Theorem~\ref{thm:online-disc} in Section~\ref{sec:pf-online-disc} and give the extension to Theorem~\ref{thm:online-disc-bern} in Section~\ref{sec:bern-ext}.

Theorems~\ref{thm:online-disc}-\ref{thm:online-disc-bern} collectively establish that up to constant factors the Bansal-Spencer algorithm is optimal within the class of online algorithms for the discrepancy minimization problem. Once again, the algorithms ruled out can succeed even with an exponentially small probability. 

At a technical level, Theorems~\ref{thm:online-disc}-\ref{thm:online-disc-bern} are established by showing the non-existence of  certain $m$-tuples of solutions described earlier with growing values of $m$, $m=\omega_M(1)$. The idea of considering the `landscape' of $m$-tuples with $m=\omega(1)$ was introduced in the context of random number partitioning problem~\cite{gamarnik2021algorithmic}. By doing
so, the authors subsequently lowered the $m$-OGP threshold and ruled out stable algorithms for a broader range of parameters than what one can get for constant $m$. Ours is the first work leveraging such a barrier with growing values of $m$ beyond stable algorithms; it further illustrates the potential gain of considering super-constant tuples for random computational problems. Another key ingredient of our proof is an anti-concentration inequality for signed sum of Bernoulli/Rademacher variables, via the Berry-Esseen Theorem. 
\section{Algorithmic Barriers in Finding Constant Discrepancy Solutions}\label{sec:const-disc}
In this section, we focus on the algorithmic problem of finding a constant discrepancy solution. More concretely, given a random $\M\in\R^{M\times n}$ we ask the following question: for what values of $M$ and $n$, can a solution $\bs\in\Sigma_n$ of constant discrepancy, $\|\M\bs\|_\infty=O(1)$, be  found efficiently? 

To begin with, a simple first-moment calculation shows that $n=\Omega(M\log M)$ is necessary for such solutions to exist. This condition turns out to be sufficient, as well; \cite{altschuler2021discrepancy} showed that if $\M$ has i.i.d.\,Bernoulli$(p)$ entries then $\D(\M)\le 1$ w.h.p.\,if $n\ge CM\log M$ where $C$ is any arbitrary constant greater than $(2\log 2)^{-1}$. On the other hand, the best known polynomial-time algorithm finding such a solution works only when $M=o(\log n)$~\cite{bansalpersonal}. This highlights a striking statistical-to-computational gap (\texttt{SCG}). 

In this section, we study the nature of this \texttt{SCG} in a smooth setting where the entries of $\M$ are i.i.d.\,$\cN(0,1)$, near the existential boundary $n=\Theta(M\log M)$.
We first focus on the `landscape' of the set of constant discrepancy solutions, and show the presence of Ensemble $m$-OGP, an intricate geometrical property. We then leverage $m$-OGP to rule out the class of stable algorithms.
\subsection{Technical Preliminaries}\label{sec:prelim}
We formalize the set of tuples of constant discrepancy solutions under investigation.
\begin{definition}\label{def:ogp-set}
    Fix a $K>0$, an $m\in\mathbb{N}$, $0<\eta<\beta<1$, and  $\mathcal{I}\subset [0,\pi/2]$. Let $\M_i\in\R^{M\times n}$, $0\le i\le m$, be i.i.d.\,random matrices, each having i.i.d.\,$\cN(0,1)$ entries. Denote by $\mathcal{S}(K,m,\beta,\eta,\mathcal{I})$ the set of all $m$-tuples $\bs_i\in\Sigma_n$, $1\le i\le m$, satisfying the following:
    \begin{itemize}
        \item {\bf (Pairwise Overlap Condition)} For any $1\le i<j\le m$,
        $
        \beta - \eta \le n^{-1}\ip{\bs_i}{\bs_j}\le \beta$.
        \item {\bf (Constant Discrepancy Condition)} There exists $\tau_1,\dots,\tau_m\in\mathcal{I}$ such that
        \\ $
        \max_{1\le i\le m}\bigl\|\M_i(\tau_i)\bs_i\bigr\|_\infty \le K$,
        where
      $\M_i(\tau_i) = \cos(\tau_i)\M_0 +\sin(\tau_i)\M_i\in\R^{M\times n}$.
    
    \end{itemize}
\end{definition}
Definition~\ref{def:ogp-set} concerns tuples of solutions of discrepancy at most $K$. The parameter $m$ is the size of tuples under consideration, and $\beta$ and $\eta$ collectively control the (forbidden) region of overlaps. Finally, the set $\mathcal{I}$ is employed for generating correlated instances; this is necessary for establishing the Ensemble $m$-OGP to rule out stable algorithms, see below for details. 
\subsection{Ensemble $m$-OGP in Discrepancy Minimization}\label{sec:ogp}
Our next main result shows that the set of constant discrepancy solutions exhibits the $m$-OGP. 
\begin{theorem}\label{thm:m-ogp-discrepancy}
    Fix arbitrary constants $C_1>c_2>0$ and a $K>0$, suppose that $C_1 M\log_2 M\ge n\ge c_2 M\log_2 M$. Then, there exists an $m\in\mathbb{N}$, a $c>0$ and $0<\eta<\beta<1$ such that the following holds. Fix any $\mathcal{I}\subset [0,\pi/2]$ with $|\mathcal{I}| \le 2^{cn}$. Then,
    $\mathbb{P}\bigl[\mathcal{S}(K,m,\beta,\eta,\mathcal{I})\ne\varnothing\bigr]\le 2^{-\Theta(n)}.
    $
\end{theorem}
 We prove Theorem~\ref{thm:m-ogp-discrepancy} in Section~\ref{sec:proof-ogp}. Several remarks are in order. Theorem~\ref{thm:m-ogp-discrepancy} shows that for any $K>0$ and throughout the entire regime $n=\Theta(M\log M)$, the set of solutions with discrepancy at most $K$ exhibits the $m$-OGP, for suitable $m,\beta$ and $\eta$. In light of prior work discussed earlier, this gives some rigorous evidence for algorithmic hardness at the boundary $n=\Theta(M\log M)$, and  constitutes  a first step towards~\cite[Conjecture~1]{altschuler2021discrepancy}. 
 
 Our proof is based on the first moment method: we show that the expected number of such $m$-tuples is exponentially small for suitably chosen $m,\beta,\eta$ and apply Markov's inequality. Further, our proof reveals that $\beta\gg \eta$: Theorem~\ref{thm:m-ogp-discrepancy} rules out $m$-tuples of constant discrepancy solutions that are nearly equidistant. Moreover, the solutions need not be of constant discrepancy w.r.t.\,the same instance: $\M_i(\tau_i)$ appearing in Definition~\ref{def:ogp-set} are potentially correlated. This is known as the Ensemble $m$-OGP and instrumental in ruling out stable algorithms in Theorem~\ref{thm:stable-hard}.
\subsection{$m$-OGP Implies Failure of Stable Algorithms}\label{sec:stable-hard}
In this section, we show that the Ensemble $m$-OGP established in Theorem~\ref{thm:m-ogp-discrepancy} implies the failure of \emph{stable algorithms} in finding a constant discrepancy solution. We begin by elaborating on the algorithmic setting and formalizing the class of stable algorithms we investigate.
\paragraph{Algorithmic Setting} An algorithm $\A$ is a mapping between $\R^{M\times n}$ and $\Sigma_n$, where randomization is allowed: we assume there exists a probability space $(\Omega,\mathbb{P}_\omega)$ such that $\A:\R^{M\times n}\times \Omega\to \Sigma_n$. For any $\omega\in\Omega$ and $\M\in\R^{M\times n}$, we want $\|\M\bs_{\rm ALG}\|_\infty = O(1)$, where $\bs_{\rm ALG} = \A(\M,\omega)\in\Sigma_n$. The class of stable algorithms is formalized as follows.
\begin{definition}\label{def:stable-alg}
    Fix a $K>0$. An algorithm $\A:\R^{M\times n}\times \Omega\to\Sigma_n$ is called $(K,\rho,p_f,p_{\rm st},f,L)$-stable (for discrepancy minimization) if it satisfies the following for all sufficiently large $M$. 
    \begin{itemize}
        \item {\bf (Success)} For $\M\in\R^{M\times n}$ with i.i.d.\,$\cN(0,1)$ entries, 
        \[
\mathbb{P}_{(\M,\omega)}\bigl[\bigl\|\M\A(\M,\omega)\bigr\|_\infty \le K\bigr]\ge 1-p_f.
        \]
        \item {\bf (Stability)} Let $\M,\overline{\M}\in\R^{M\times n}$ be random matrices, each with i.i.d.\,$\cN(0,1)$ entries, such that $\mathbb{E}\bigl[\M_{ij}\overline{\M}_{ij}\bigr]=\rho$ for $1\le i\le M$ and $1\le j\le n$. Then,
        \[
\mathbb{P}_{(\M,\overline{M},\omega)}\Bigl[d_H\bigl(\A(\M,\omega),\A(\overline{\M},\omega)\bigr)\le f+L\|\M-\overline{\M}\|_F\Bigr]\ge 1-p_{\rm st}.
        \]
    \end{itemize}
\end{definition}
Definition~\ref{def:stable-alg} is the same as \cite[Definition~3.1]{gamarnik2022algorithms}. W.p.\,at least $1-p_f$, $\A$ finds a solution of discrepancy below $K$. $\A$ can tolerate an input correlation value of $\rho$; and the parameters $f$ and $L$ quantify the sensitivity of the output of $\A$ to changes in its input. The stability guarantee is probabilistic—w.r.t.\,both $\M,\overline{\M}$ and to the randomness $\omega$ of $\A$—holding w.p.\,at least $1-p_{\rm st}$. Finally, the term $f$ makes our negative result only stronger: $\A$ is allowed to make $f$ flips even when $\M$ and $\overline{\M}$ are `too close’. Our final main result is as follows.
\begin{theorem}\label{thm:stable-hard}
    Fix a $K>0$, $C_1>c_2>0$ and a $\mathcal{L}>0$. Suppose $
    C_1 M\log_2 M\ge n\ge c_2 M\log_2 M$. Let $m\in\mathbb{N}$ and $0<\eta<\beta<1$ be the $m$-OGP parameters prescribed by Theorem~\ref{thm:m-ogp-discrepancy}. Set
   \begin{equation}\label{eq:C-Q-T}
        C=\frac{\eta^2}{1600},\quad Q=\frac{4800\mathcal{L}\pi}{\eta^2},\quad\text{and}\quad T=\exp_2\left(2^{4mQ\log_2 Q}\right).
    \end{equation}
    Then, there exists an $n_0\in\mathbb{N}$ such that the following holds. For every $n\ge n_0$, there exists no randomized algorithm $\A:\R^{M\times n}\times\Omega\to\Sigma_n$ which, in the sense of Definition~\ref{def:stable-alg}, is
    \[
\left(K,\cos\left(\frac{\pi}{2Q}\right),\frac{1}{9(Q+1)T},\frac{1}{9Q(T+1)},Cn,\mathcal{L}\sqrt{\frac{n}{M}}\right)-\text{stable}.
    \]
\end{theorem}
The proof of Theorem~\ref{thm:stable-hard} is almost identical to that of~\cite[Theorem~3.2]{gamarnik2022algorithms}, and omitted for brevity. Several remarks are in order. 

Firstly, there is no restriction on the running time of $\A$: Theorem~\ref{thm:stable-hard} rules out any $\A$ that is stable with suitable parameters in the sense of Definition~\ref{def:stable-alg}. Secondly, observe that $m,\beta,\eta,\mathcal{L}$ are all $O(1)$ (as $n\to\infty$). Hence, $C,Q,T$ per~\eqref{eq:C-Q-T} are all $O(1)$, as well. This is an important feature of our result: Theorem~\ref{thm:stable-hard} rules out algorithms with a constant success/stability guarantee. Lastly, since $C=O(1)$,  $\A$ is still allowed to make $\Theta(n)$ bit flips even when $\M$ and $\overline{\M}$ are nearly identical.

\section{Proofs}\label{sec:pfs}
\paragraph{Additional Notation} We commence this section with an additional list of notation. For any set $A$, $|A|$ denotes its cardinality. Given any event $E$, denote its indicator by $\ind\{E\}$. For any $v=(v(i):1\le i\le n)\in\R^n$ and $p>0$, $\|v\|_p  =  \bigl(\sum_{1\le i\le n}|v(i)|^p\bigr)^{1/p}$ and $\|v\|_\infty = \max_{1\le i\le n}|v(i)|$. For $v,v'\in\R^n$, $\ip{v}{v'}\triangleq  \sum_{1\le i\le n}v(i)v'(i)$. For any $\bs,\bs'\in \Sigma_n\triangleq \{-1,1\}^n$, $d_H(\bs,\bs')$ denotes their Hamming distance: $d_H(\bs,\bs') \triangleq \sum_{1\le i\le n}\ind\{\bs(i)\ne \bs'(i)\}$. For any $r>0$, $\log_r(\cdot)$ and $\exp_r(\cdot)$ denote, respectively, the logarithm and exponential functions base $r$; when $r=e$, we omit the subscript. For $p\in[0,1]$, $h_b(p) \triangleq -p\log_2 p-(1-p)\log_2(1-p)$. Denote by $I_k$ the $k\times k$ identity matrix, and by $\boldsymbol{1}$ the vector of all ones whose dimension will be clear from the context. Given $\boldsymbol{\mu}\in\R^k$ and $\Sigma\in\R^{k\times k}$, denote by $\cN(\boldsymbol{\mu},\Sigma)$ the multivariate normal distribution in $\R^k$ with mean $\boldsymbol{\mu}$ and covariance $\Sigma$. Given a matrix $\M$,  $\|\M\|_F$, $\|\M\|_2$, and $|\M|$ denote, respectively, the Frobenius norm, the spectral norm, and the determinant of $\M$. 

We employ standard Bachmann-Landau asymptotic notation throughout, e.g.\,$\Theta(\cdot)$, $O(\cdot)$, $o(\cdot)$, $\Omega(\cdot)$, where the underlying asymptotics will often be clear from the context. In certain cases where a confusion is possible, we reflect the underlying asymptotics as a subscript, e.g.\,$\Theta_\kappa(\cdot)$. All floor/ceiling operators are omitted for the sake of simplicity. 
\subsection{Proof of Theorem~\ref{thm:online-sbp}}\label{sec:pf-online-sbp}
Let $\kappa,\alpha>0$, $M=n\alpha$, $m\in\mathbb{N}$, and $\Delta\in(0,\frac12)$. Suppose $\M_1\in\R^{M\times n}$ has i.i.d.\,$\cN(0,1)$ entries and let $\M_2,\dots,\M_m\in\R^{M\times n}$ be random matrices obtained from $\M_1$ by independently resampling the last $\Delta n$ columns of $\M_1$. Denote by $\Xi(m,\Delta)$ the set of all $m$-tuples satisfying the following:
\begin{itemize}
    \item $\max_{1\le i\le m}\bigl\|\M_i\bs_i\bigr\|_\infty \le \kappa\sqrt{n}$.
    \item For $1\le i<j\le n$ and $1\le k\le n-\Delta n$, $\bs_i(k)=\bs_j(k)$.
\end{itemize}
We establish the following proposition.
\begin{proposition}\label{prop:sbp}
Fix any $\kappa>0$ small enough and let $\alpha\ge 4\kappa^2$. Then, there exists an $m\in\mathbb{N}$ and a $\Delta\in(0,1/2)$ such that 
\[
\mathbb{P}\bigl[\Xi(m,\Delta)=\varnothing\bigr]\ge 1-e^{-\Theta(n)}.
\]
\end{proposition}
We first assume Proposition~\ref{prop:sbp} and show how to deduce Theorem~\ref{thm:online-sbp}. Fix a $c>0$ and suppose, for the sake of contradiction, that an $(e^{-cn},\alpha)$-online $\A$ exists. For $\M_1,\dots,\M_m\in\R^{M\times n}$ described above, set  
\begin{equation}\label{eq:sigma-i-online}
\bs_i \triangleq \A(\M_i)\in\Sigma_n,\quad 1\le i\le m.
\end{equation}
Note that for any $1\le i<j\le m$, the first $n-\Delta n$ columns of $\M_i$ and $\M_j$ are identical. Consequently, 
\[
\bs_i(k) = \bs_j(k)\quad\text{for}\quad 1\le i<j\le m\quad\text{and}\quad 1\le k\le n-\Delta n.
\]
Next, we establish the following probability guarantee.
\begin{lemma}\label{lemma:jensen-sbp}
    \[
\mathbb{P}\left[\max_{1\le i\le m}\bigl\|\M_i\bs_i\bigr\|_\infty\le \kappa\sqrt{n}\right]\ge p_s^m.
    \]
\end{lemma}
\begin{proof}[Proof of Lemma~\ref{lemma:jensen-sbp}]
Our argument is based on a clever application of Jensen's inequality, due to~\cite[Lemma~5.3]{gamarnik2017performance}. Denote by $\zeta$ the first $(1-\Delta)n$ columns of $\M_1$. That is, $\zeta$ is the `common randomness' shared by $\M_1,\dots,\M_m$. Set
\[
I_i = \ind\bigl\{\|\M_i\bs_i\|_\infty\le \kappa\sqrt{n}\bigr\}.
\]
Then,
\[
 \mathbb{P}\left[\max_{1\le i\le m}\bigl\|\M_i\sigma_i\bigr\|_\infty\le \kappa\sqrt{n}\right]= \mathbb{E}\bigl[I_1\cdots I_m\bigr].
\]
We then complete the proof by noticing
\begin{align*}
    \mathbb{E}\bigl[I_1\cdots I_m\bigr] =\mathbb{E}_{\zeta}\Bigl[\mathbb{E}\bigl[I_1\cdots I_m|\zeta\bigr]\Bigr]=\mathbb{E}_{\zeta}\left[\mathbb{E}[I_1|\zeta]^m\right]\ge \left(\mathbb{E}_{\zeta}\bigl[\mathbb{E}[I_1|\zeta]\bigr]\right)^m=\mathbb{E}[I_1]^m = p_s^m,
\end{align*}
where we used fact that $I_1,\dots,I_m$ are independent conditional on $\zeta$ and Jensen's inequality.
\end{proof}
Note that clearly $(\bs_1,\dots,\bs_m)\in \Xi(m,\Delta)$, in particular $\Xi(m,\Delta)$ is non-empty w.p.\,at least $e^{-cmn}$. Using Proposition~\ref{prop:sbp}, we obtain
\[
e^{-\Theta(n)}\ge \mathbb{P}\bigl[\Xi(m,\Delta)\ne\varnothing]\ge e^{-cmn}.
\]
If $m$ is constant and $c>0$ is sufficiently small, this is a contradiction for all large enough $n$. Therefore, it suffices to establish Proposition~\ref{prop:sbp}.
\begin{proof}[Proof of Proposition~\ref{prop:sbp}]
Our proof is based on the \emph{first moment method}. Let
\[
\mathcal{S} = \bigl\{(\bs_1,\dots,\bs_m):\bs_i(k)=\bs_j(k),1\le i<j\le m,1\le k\le n-\Delta n\bigr\}.
\]
Observe that
\begin{equation}\label{eq:auxil-rv}
    \bigl|\Xi(m,\Delta)\bigr| = \sum_{(\bs_1,\dots,\bs_m)\in\mathcal{S}}\ind\left\{\max_{1\le i\le m}\|\M_i\bs_i\|_\infty\le \kappa\sqrt{n}\right\}.
\end{equation}
In what follows, we show that for a suitable $m\in\mathbb{N}$ and $\Delta\in(0,1/2)$, 
\[
\mathbb{E}\bigl[|\Xi(m,\Delta)|\bigr]\le e^{-\Theta(n)}.
\]
\paragraph{Counting Estimate} We bound $|\mathcal{S}|$. There are $2^n$ choices for $\bs_1\in\Sigma_n$. Having chosen a $\bs_1$, there are $2^{\Delta n}$ choices for any $\bs_i$, $2\le i\le m$. So,
\begin{equation}\label{eq:card-of-S-online}
    |\mathcal{S}| \le 2^n \bigl(2^{\Delta n}\bigr)^{m-1}\le \exp_2\bigl(n+nm\Delta\bigr)
\end{equation} 
\paragraph{Probability Estimate} Fix any $(\bs_1,\dots,\bs_m)\in\mathcal{S}$. Denote by $R_1,\dots,R_m\in\R^n$ the first rows of $\M_1,\dots,\M_m$, respectively; and set
\[
Z_i = \frac{1}{\sqrt{n}}\ip{R_i}{\bs_i} \distr \cN(0,1),\quad 1\le i\le m.
\]
Observe that if $k\ne k'$ or $n-\Delta n+1\le k=k'\le n$, $\mathbb{E}\bigl[R_i(k)R_j(k')\bigr] = 0$. Using this fact, we immediately conclude that $\mathbb{E}[Z_iZ_j]=1-\Delta$. In particular, $(Z_1,\dots,Z_m)\in \R^m$ is a centered multivariate normal random vector with covariance $\Sigma$, where
\[
\Sigma = \Delta I_m + (1-\Delta) \boldsymbol{1}\boldsymbol{1}^T \in \R^{m\times m},
\]
where $\boldsymbol{1}\in\R^m$ is the vector of all ones.  
In particular, the spectrum of $\Sigma$ consists of the eigenvalue $\Delta+(1-\Delta)m$ with multiplicity one and the eigenvalue $\Delta$ with multiplicity $m-1$. We then obtain
\begin{align}
    \mathbb{P}\left[\max_{1\le i\le m}\|\M_i\bs_i\|_\infty\le \kappa\sqrt{n}\right]&\le \mathbb{P}\left[\max_{1\le i\le m}\bigl|\ip{R_i}{\bs_i}\bigr|\le \kappa\sqrt{n}\right]^{\alpha n} \nonumber \\
    &\left((2\pi)^{-\frac{m}{2}}|\Sigma|^{-\frac12}\int_{\boldsymbol{z}\in[-\kappa,\kappa]^m}\exp\left(-\frac{\boldsymbol{z}^T \Sigma^{-1}\boldsymbol{z}}{2}\right)\right)^{\alpha n}\nonumber \\ 
    &\le\left((2\pi)^{-\frac{m}{2}}\bigl(\Delta+(1-\Delta)m\bigr)^{-\frac12}\Delta^{-\frac{m-1}{2}}(2\kappa)^m\right)^{\alpha n}\label{eq:online-prob}.
\end{align}
\paragraph{Estimating $\mathbb{E}[|\Xi(\Delta,m)|]$} We now combine~\eqref{eq:card-of-S-online} and~\eqref{eq:online-prob} to arrive at
\begin{equation}\label{eq:psi}
    \mathbb{E}\bigl[\bigl|\Xi(\Delta,m)\bigr|\bigr]\le \exp_2\Bigl(n\Psi(\Delta,m,\alpha)\Bigr),
\end{equation}
where
\begin{align*}
    \Psi(\Delta,m,\alpha)&=1+m\Delta - \frac{\alpha m}{2}\log_2(2\pi)+\alpha m\log_2(2\kappa)-\frac{\alpha(m-1)}{2}\log_2\Delta-\frac{\alpha}{2}\log_2\bigl(\Delta+(1-\Delta)m\bigr).
    \end{align*}
Using the fact $\log_2 \frac1\Delta>0$ if $\Delta<\frac12$, we further arrive at the bound
\begin{equation}\label{eq:psi-up}
     \Psi(\Delta,m,\alpha)\le m\left(\frac1m-\frac{\alpha}{2m}\log_2\bigl(\Delta+(1-\Delta)m\bigr)+\Upsilon(\Delta,\alpha)\right),
\end{equation}
for
\[
\Upsilon(\Delta,\alpha) = \Delta -\frac{\alpha}{2}\log_2(2\pi)+\alpha \log_2(2\kappa) - \frac{\alpha}{2}\log_2\Delta.
\]
\paragraph{Analyzing $\Upsilon(\Delta,\alpha)$} We set $\Delta=(2\kappa)^2$, so that
\[
\alpha\log_2 (2\kappa) - \frac{\alpha}{2}\log_2 \Delta=0.
\]
Next, fix any $\alpha\ge 4\kappa^2$. Then,
\begin{equation}\label{eq:ups-up}
    \Upsilon(\Delta,\alpha) = \Delta -\frac{\alpha}{2}\log_2(2\pi) \le 4\kappa^2 -2\kappa^2\log_2(2\pi) = -\kappa^2\bigl(2\log_2(2\pi)-4\bigr) = -\Theta_\kappa(\kappa^2).
\end{equation}
\paragraph{Combining everything} For fixed small $\kappa>0$, $\alpha\ge 4\kappa^2$, and $\Delta=(2\kappa)^2$; we have $\Upsilon(\Delta,\alpha)=-\Theta_\kappa(\kappa^2)<0$. Furthermore, 
\[
\frac1m -\frac{\alpha}{2m}\log_2\bigl(\Delta+(1-\Delta)m\bigr) = o_m(1)
\]
as $m\to\infty$. Note that $\Upsilon(\Delta,\alpha)$ depends only on $\alpha,\kappa$. So, for $m\in\mathbb{N}$ sufficiently large,~\eqref{eq:ups-up} yields
\[
\frac1m-\frac{\alpha}{2m}\log_2\bigl(\Delta+(1-\Delta)m\bigr) +\Upsilon(\Delta,\alpha)<0.
\]
Hence, combining~\eqref{eq:psi} and~\eqref{eq:psi-up}, we get
\[
\mathbb{E}\bigl[\bigl|\Xi(m,\Delta)\bigr|\bigr] \le e^{-\Theta(n)}. 
\]
From here, we conclude by Markov's inequality as
\[
\mathbb{P}\bigl[\bigl|\Xi(\Delta,m)\bigr|\ge 1\bigr]\le \mathbb{E}\bigl[\bigl|\Xi(\Delta,m)\bigr|\bigr]=\exp\bigl(-\Theta(n)\bigr).
\]
\end{proof}
\subsection{Proof of Theorem~\ref{thm:online-disc}}\label{sec:pf-online-disc}
The proof of Theorem~\ref{thm:online-disc} is similar to that of Theorem~\ref{thm:online-sbp}. We first establish the following proposition. 
\begin{proposition}\label{prop:online-disc}
Let $n\ge M=\omega(1)$, $\M_1\in\{-1,1\}^{M\times n}$ with i.i.d.\,Rademacher entries and $m=\lceil \frac{2n}{M} \rceil$. Generate $\M_2,\dots,\M_m\in\{-1,1\}^{M\times n}$ by independently resampling the last $M$ columns of $\M_1$. Denote by $\Xi_d(m,M)$ the set of all $m$-tuples $\bs_1,\dots,\bs_m\in\Sigma_n$ satisfying the following:
\begin{itemize}
    \item $\max_{1\le i\le m}\|\M_i\bs_i\|\le C_u\sqrt{M}$, where $C_u\triangleq \frac{1}{24}$.
    \item For $1\le i<j\le m$ and $1\le k\le n-M$, $\bs_i(k)=\bs_j(k)$. 
\end{itemize}
Then, 
\[
\mathbb{P}\bigl[\Xi_d(m,M)=\varnothing\bigr]\ge 1-e^{-n}.
\]
\end{proposition}
Before proving Proposition~\ref{prop:online-disc}, we highlight that if $n=\omega(M)$ then $m=\omega_M(1)$ and the fraction $\Delta = M/n$ of the resampled columns is vanishing. We first show how Proposition~\ref{prop:online-disc} yields Theorem~\ref{thm:online-disc}. Suppose, for the sake of contradiction, that an $\A:\{-1,1\}^{M\times n}\to\Sigma_n$ which is $(e^{-cM},\sqrt{M}/24)$-optimal (with $c<1/2$ arbitrary) exists. For $\M_i$, $1\le i\le m$, as in the proposition, define 
\[
\bs_i\triangleq \A(\M_i)\in\Sigma_n,\quad 1\le i\le n
\]
and observe that
\[
\bs_i(k)=\bs_j(k),\quad\text{for all}\quad 1\le i<j\le m\quad\text{and}\quad 1\le k\le n-M
\]
as $\A$ is online. We then establish
\begin{lemma}\label{lemma:online-disc}
    \[
    \mathbb{P}\left[\max_{1\le i\le m}\|\M_i\bs_i\|\le \frac{\sqrt{M}}{24}\right]\ge \left(e^{-cM}\right)^m \geq e^{-2cn}.
    \]
\end{lemma}
Proof of Lemma~\ref{lemma:online-disc} is identical to Lemma~\ref{lemma:jensen-sbp}. So, under the assumption that such an $\A$ exists, we obtain $\Xi_d(m,\Delta)\ne\varnothing$ w.p.\,at least $e^{-2cn}$. Finally using Proposition~\ref{prop:online-disc},
\[
e^{-n}\ge \mathbb{P}\bigl[\Xi_d(m,M)\ne\varnothing]\ge e^{-2cn}
\]
which is a contradiction since $c<1/2$. Hence, it suffices to establish Proposition~\ref{prop:online-disc}. 

\begin{proof}[Proof of Proposition~\ref{prop:online-disc}]
The proof of Proposition~\ref{prop:online-disc} is similar to Proposition~\ref{prop:sbp}; it is based in particular on the first moment method. Let
\[
\bar{\mathcal{S}} = \bigl\{(\bs_1,\dots,\bs_m):\bs_i(k)=\bs_j(k),1\le i<j\le m,1\le k\le n-M\bigr\}.
\]
Note that
\begin{equation}\label{eq:first-mom}
    \bigl|\Xi_d(m,M)\bigr|= \sum_{(\sigma_1,\dots,\sigma_m)\in \bar{\mathcal{S}}}\ind\left\{\max_{1\le i\le m}\|\M_i\bs_i\|_\infty \le C_u\sqrt{M}\right\},\quad\text{where}\quad C_u = \frac{1}{24}.
\end{equation}
\paragraph{Counting term} We bound $|\bar{\mathcal{S}}|$. There are $2^n$ choices for $\bs_1$ and having fixed it, there are $2^M$ choices for any $\bs_i$, $2\le i\le m$. So,
\begin{equation}\label{eq:cardinality-bd}
|\bar{\mathcal{S}}|\le 2^n (2^M)^{m-1} \le \exp_2\bigl(n+mM\bigr).
\end{equation}
\paragraph{Probability term.} 
Fix an arbitrary $(\bs_1,\dots,\bs_m)\in\bar{\mathcal{S}}$. Let $R_i\in\{\pm 1\}^n$, $1\le i\le m$, denote respectively the first rows of $\M_i$, $1\le i\le m$. For each fixed $i$, the rows of $\M_i$ are independent. So,
\begin{equation}\label{eq:prob-term1}
    \mathbb{P}\left[\max_{1\le i\le m}\|\M_i\bs_i\|_\infty \le C_u\sqrt{M}\right] = \mathbb{P}\left[\max_{1\le i\le m}\left|\ip{R_i}{\bs_i}\right|\le C_u\sqrt{M}\right]^M.
\end{equation}
Next, let 
\[
R_i=\bigl(R_{ik}:1\le k\le n\bigr),\quad 1\le i\le m.
\]
Fix any $1\le i<j\le m$. Observe that the random vectors
\[
\bigl(R_{ik}:1\le k\le n-M\bigr)\quad\text{and}\quad \bigl(R_{jk}:1\le k\le n-M\bigr)
\]
are identical. For this reason, we drop the first index and use $\bigl(R_k:1\le k\le n-M\bigr)$ instead. Next, fix any $\boldsymbol{v}=(v_1,\dots,v_{n-M})\in\{-1,1\}^{n-M}$ and define
\begin{equation}\label{eq:Delta-i-and-sigma-i}
    \Delta_i(\boldsymbol{v})\triangleq \sum_{1\le k\le n-M}v_k\bs_i(k)\qquad\text{and}\qquad \Sigma_i \triangleq \sum_{n-M+1\le k\le n}R_{ik}\bs_i(k).
\end{equation}
Our goal is to control the right hand side in~\eqref{eq:prob-term1}. To that end, our strategy is to condition on $R_1,\dots,R_{n-M}$ and apply Berry-Esseen inequality for $\Sigma_i$. 
We establish the following auxiliary result.
\begin{lemma}\label{lemma:berry}
Let $Z_1,\dots,Z_M$ be i.i.d.\,Rademacher random variables, $\epsilon_i\in\{-1,1\}$, $1\le i\le M$, be deterministic signs, and $I\subset \R$ be an interval of length $|I|=\omega_M(1)$. Then
\[
\mathbb{P}\bigl[Z_1\epsilon_1+\cdots+Z_M \epsilon_M\in I\bigr] \le \frac{3|I|}{\sqrt{M}}
\]
for every large enough $M$.
\end{lemma}
\begin{proof}[Proof of Lemma~\ref{lemma:berry}]
Let $\frac{1}{\sqrt{M}}I$ denotes the set $\{c/\sqrt{M}:c\in I\}$.  By the Central Limit Theorem,
\[
\frac{1}{\sqrt{M}}\sum_{1\le i\le M}Z_i\epsilon_i \Rightarrow \mathcal{N}(0,1)
\]
in distribution, where the speed of convergence is controlled by the Berry-Esseen inequality:
\begin{equation}\label{eq:be1}
    \left|\mathbb{P}\left[\sum_{1\le i\le M}Z_i \epsilon_i \in I\right]-\mathbb{P}\left[\mathcal{N}(0,1)\in \frac{1}{\sqrt{M}}I\right]\right|\le \frac{\mathcal{C}_{\rm be}}{\sqrt{M}}.
\end{equation}
Here, $\mathcal{C}_{\rm be}>0$ is an absolute constant. Furthermore, we have
\begin{equation}\label{eq:be2}
    \mathbb{P}\left[\mathcal{N}(0,1)\in \frac{1}{\sqrt{M}}I\right] = \frac{1}{\sqrt{2\pi}}\int_{u\in \frac{1}{\sqrt{M}}I}\exp(-u^2/2)\;du \le \frac{|I|}{\sqrt{2\pi M}}.
\end{equation}
Combining~\eqref{eq:be1} and~\eqref{eq:be2} via triangle inequality, we obtain that for all large enough $M$,
\begin{equation}\label{eq:be3}
    \mathbb{P}\left[\sum_{1\le i\le M}Z_i \epsilon_i \in I\right] \le \frac{1}{\sqrt{M}}\left(\mathcal{C}_{\rm be}+\frac{|I|}{\sqrt{2\pi}}\right) \le \frac{3|I|}{\sqrt{M}},
\end{equation}
where we recalled $|I|=\omega_M(1)$ and $\mathcal{C}_{\rm be}=O_M(1)$. This establishes Lemma~\ref{lemma:berry}.
\end{proof}
Next, fix any $\boldsymbol{v}\in\{-1,1\}^{n-M}$ and set
\[
I_i(\boldsymbol{v}) = \left[-C_u\sqrt{M}-\Delta_i(\boldsymbol{v}),C_u\sqrt{M} - \Delta_i(\boldsymbol{v})\right],
\]
where we recall $\Delta_i(\boldsymbol{v})$ from~\eqref{eq:Delta-i-and-sigma-i}. In particular, 
\begin{equation}\label{eq:size-of-Ii}
    \bigl|I_i(\boldsymbol{v})\bigr| =  2C_u\sqrt{M},\quad\text{for all}\quad 1\le i\le m\quad\text{and}\quad v\in\{-1,1\}^{n-M}.
\end{equation}
Next fix a $1\le i\le m$ and recall $\Sigma_i$ per~\eqref{eq:Delta-i-and-sigma-i}. Applying Lemma~\ref{lemma:berry}, we conclude that 
\begin{equation}\label{eq:be4}
    \max_{1\le i\le m}  \max_{\boldsymbol{v}\in\{-1,1\}^{n-M}}
\mathbb{P}\Bigl[\Sigma_i\in I_i(\boldsymbol{v})\Bigr]\le 6C_u.
\end{equation}
We are ready to bound the probability term~\eqref{eq:prob-term1} by conditioning on $R_1,\dots,R_{n-M}$. 
\begin{align}
    &\mathbb{P}\left[\max_{1\le i\le m}\left|\ip{R_i}{\bs_i}\right|\le C_u\sqrt{M}\right] \nonumber \\
    &= \sum_{\boldsymbol{v}\in\{-1,1\}^{n-M}}\mathbb{P}\Bigl[\Sigma_i\in I_i(\boldsymbol{v}),1\le i\le m\Big| (R_1,\dots,R_{n-M})=\boldsymbol{v}\Bigr]\underbrace{\mathbb{P}\left[(R_1,\dots,R_{n-M})=\boldsymbol{v}\right]}_{=2^{-(n-M)}} \label{eq:condition}\\
    &=2^{-(n-M)}\sum_{\boldsymbol{v}\in\{-1,1\}^{n-M}}\mathbb{P}\Bigl[\Sigma_i\in I_i(\boldsymbol{v}),1\le i\le m \Bigr] \label{eq:independence}\\
    &=2^{-(n-M)}\sum_{v\in\{-1,1\}^{n-M}}\prod_{1\le i\le m}\mathbb{P}\bigl[\Sigma_i\in I_i(\boldsymbol{v})\bigr]\label{eq:independence2} \\
    &\le (6C_u)^m\label{eq:lemmaberry}.
\end{align}
We now justify the lines above. Equation~\eqref{eq:condition} follows by conditioning on the `common randomness' $R_1,\dots,R_{n-M}$ and recalling that they are uniform over $\{-1,1\}^{n-M}$. Equation~\eqref{eq:independence} uses the fact for any fixed $1\le i\le m$, $\Sigma_i$ is independent of $R_1,\dots,R_{n-M}$, and~\eqref{eq:independence2} uses the fact $\Sigma_1,\dots,\Sigma_m$ is also a collection of independent random variables. Finally,~\eqref{eq:lemmaberry} uses~\eqref{eq:be4}. 

Combining~\eqref{eq:prob-term1} with~\eqref{eq:lemmaberry}, we thus conclude
\begin{equation}\label{eq:prob-main}
    \max_{(\bs_1,\dots,\bs_m)\in \bar{\mathcal{S}}}\mathbb{P}\left[\max_{1\le i\le m}\|\M_i\bs_i\|_\infty\le C_u\sqrt{M}\right]\le (6C_u)^{mM}.
\end{equation}
\paragraph{Bounding $\mathbb{E}\bigl[\bigl|\Xi_d(m,M)\bigr|\bigr]$} We are ready to estimate $\mathbb{E}\bigl[\bigl|\Xi_d(m,M)\bigr|\bigr]$. Using~\eqref{eq:first-mom},~\eqref{eq:cardinality-bd} and~\eqref{eq:prob-main}, 
\[
\mathbb{E}\bigl[\bigl|\Xi_d(m,M)\bigr|\bigr]\le \exp_2\left(n+mM-mM\log_2\frac{1}{6C_u}\right).
\]
Inserting the values $C_u = 1/24$ and $m \geq 2n/M$, we obtain
\[
n+mM - mM\log_2\frac{1}{6C_u} \leq -n,
\]
so that $\mathbb{E}\bigl[\bigl|\Xi_d(m,M)\bigr|\bigr]\le e^{-n}$. Finally, we conclude by applying Markov's inequality:
\[
\mathbb{P}\bigl[\Xi_d(m,M)\ne\varnothing\bigr]=\mathbb{P}\bigl[\bigl|\Xi_d(m,M)\bigr|\ge 1\bigr]\le \mathbb{E}\bigl[\bigl|\Xi_d(m,M)\bigr|\bigr]\le e^{-n}. 
\]
\end{proof}
\subsection{Proof Sketch for Theorem~\ref{thm:online-disc-bern}}\label{sec:bern-ext}
The proof of Theorem~\ref{thm:online-disc-bern} is quite similar to Theorem~\ref{thm:online-disc}; we only highlight the necessary changes.
Let $\M_1$ consists of i.i.d.\,Bernoulli$(p)$ entries. 
Suppose that there exists an $\A:\R^{M\times n}\to \Sigma_n$ that is $(e^{-cM},C_u'\sqrt{M})$-online  in the sense of Definition~\ref{def:online-algs}, where $c<1/2$ is arbitrary and
\[
C_u' = \frac{\sqrt{p-p^2}}{24}.
\]
We set $m=2n/M$ and show how to adapt Proposition~\ref{prop:online-disc} to this case. Once this is done, the rest follows verbatim from Theorem~\ref{thm:online-disc}. First, all instances of $C_u$ in the proof of Theorem~\ref{thm:online-disc} are replaced with $C_u' = C_u\sqrt{p-p^2}$. Next, the counting estimate per~\eqref{eq:cardinality-bd} remains intact. Lemma~\ref{lemma:berry}, on the other hand, is replaced with the following.
\begin{lemma}\label{lemma:berry2}
    Let $Z_1,\dots,Z_M$ be i.i.d.\,Bernoulli$(p)$ random variables, $\epsilon_i\in\{-1,1\}$, $1\le i\le M$, be deterministic signs, and $I\subset \R$ be an interval of length $|I|=\omega_{(p-p^2)M}(1)$. Then
    \[
\mathbb{P}\bigl[Z_1\epsilon_1 + \cdots + Z_M\epsilon_M\in I\bigr]
     \le \frac{3|I|}{\sqrt{M(p-p^2)}},
     \]
     for every large enough $M$.
\end{lemma}
\begin{proof}[Proof of Lemma~\ref{lemma:berry2}]
    Observe that $\mathbb{E}[Z_i\epsilon_i] = p\epsilon_i$ and \[
    {\rm Var}(Z_i\epsilon_i) = \mathbb{E}[Z_i^2\epsilon_i^2] - p^2\epsilon_i^2 = p-p^2,
    \]
    as $\epsilon_i\in\{-1,1\}$. Thus by the CLT,
    \[
    \frac{1}{\sqrt{(p-p^2)M}}\left(\sum_{1\le i\le M}Z_i\epsilon_i - p\ip{\boldsymbol{1}}{\boldsymbol{\epsilon}}\right)\Rightarrow \cN(0,1)
    \]
    in distribution, where $\boldsymbol{1}\in\R^M$ is the vector of all ones and $\boldsymbol{\epsilon}=(\epsilon_i:1\le i\le M)\in\{-1,1\}^M$. Further, by the Berry-Esseen inequality, we have that
    \[
    \left|\mathbb{P}\left[\sum_{1\le i\le M}Z_i\epsilon_i \in I\right] - \mathbb{P}\left[\cN(0,1)\in \frac{I-p\ip{\boldsymbol{1}}{\boldsymbol{\epsilon}}}{\sqrt{(p-p^2)M}}\right]\right|\le \frac{\mathcal{C}'_{\rm be}}{\sqrt{(p-p^2)M}}
    \]
    for some absolute constant $\mathcal{C}'_{\rm be}>0$. Here,
    \[
    \frac{I-p\ip{\boldsymbol{1}}{\boldsymbol{\epsilon}}}{\sqrt{(p-p^2)M}} = \left\{\frac{c-p\ip{\boldsymbol{1}}{\boldsymbol{\epsilon}}}{\sqrt{(p-p^2)M}}:c\in I\right\},
    \]
    so that
    \[
    \left| \frac{I-p\ip{\boldsymbol{1}}{\boldsymbol{\epsilon}}}{\sqrt{(p-p^2)M}}\right| = \frac{|I|}{\sqrt{(p-p^2)M}},
    \]
    using the translation invariance of Lebesgue measure. From here, proceeding in the exact same way as in the proof of Lemma~\ref{lemma:berry},
    we establish Lemma~\ref{lemma:berry2}. 
\end{proof}
Equipped with Lemma~\ref{lemma:berry2} and using the exact same notation,~\eqref{eq:be4} modifies to~\eqref{eq:be5} where 
\begin{equation}\label{eq:be5}
    \max_{1\le i\le m}  \max_{\boldsymbol{v}\in\{-1,1\}^{n-M}}
\mathbb{P}\Bigl[\Sigma_i\in I_i(\boldsymbol{v})\Bigr]\le \frac{6C_u'}{\sqrt{p-p^2}} = \frac14.
\end{equation}
We now proceed analogously to lines~\eqref{eq:condition}-\eqref{eq:lemmaberry}. Note that for any arbitrary $\boldsymbol{v}\in\{0,1\}^{n-M}$, 
\begin{equation}\label{eq:be6}
\mathbb{P}\Bigl[\Sigma_i\in I_i(\boldsymbol{v}),1\le i\le m\Big| (R_1,\dots,R_{n-M})=\boldsymbol{v}\Bigr]\le 2^{-2m},
\end{equation}
using~\eqref{eq:be5}. Hence,
\begin{align*}
    &\mathbb{P}\left[\max_{1\le i\le m}\left|\ip{R_i}{\bs_i}\right|\le C_u'\sqrt{M}\right] \\
    &\le \sum_{\boldsymbol{v}\in\{0,1\}^{n-M}} \mathbb{P}\Bigl[\Sigma_i\in I_i(\boldsymbol{v}),1\le i\le m\Big| (R_1,\dots,R_{n-M})=\boldsymbol{v}\Bigr]\mathbb{P}\Bigl[(R_1,\dots,R_{n-M})=\boldsymbol{v}\Bigr] \\
    &\le 2^{-2m}\sum_{\boldsymbol{v}\in\{0,1\}^{n-M}} \mathbb{P}\Bigl[(R_1,\dots,R_{n-M})=\boldsymbol{v}\Bigr]\\
    &=2^{-2m},
\end{align*}
where we used~\eqref{eq:be6} in the penultimate line. This is precisely the same bound as~\eqref{eq:lemmaberry}, so the rest of the proof remains intact. This completes the proof of Theorem~\ref{thm:online-disc-bern}.
\subsection{Proof of Theorem~\ref{thm:m-ogp-discrepancy}}\label{sec:proof-ogp}
Fix a $K>0$, $C_1>c_2>0$, and suppose
\begin{equation}\label{eq:n-scale}
    C_1 M\log_2 M\ge n\ge c_2 M\log_2 M. 
\end{equation}
We establish our result via the first-moment method. Notice that by Markov's inequality,
\[
\mathbb{P}\bigl[\mathcal{S}(K,m,\beta,\eta,\mathcal{I})\ne\varnothing\bigr] = 
\mathbb{P}\bigl[\bigl|\mathcal{S}(K,m,\beta,\eta,\mathcal{I})\bigr|\ge 1\bigr] \le 
\mathbb{E}\bigl[\bigl|\mathcal{S}(K,m,\beta,\eta,\mathcal{I})\bigr|\bigr].
\]
So, it suffices to prove that
\[
\mathbb{E}\bigl[\bigl|\mathcal{S}(K,m,\beta,\eta,\mathcal{I})\bigr|\bigr]\le 2^{-\Theta(n)}.
\]
We now estimate $\mathbb{E}\bigl[\bigl|\mathcal{S}(K,m,\beta,\eta,\mathcal{I})\bigr|\bigr]$. 
\paragraph{Counting term} Fix $m\in\mathbb{N}$, $0<\eta<\beta<1$ and denote by $M(m,\beta,\eta)$ the number of $m$-tuples $(\bs_i\in\Sigma_n:1\le i\le m)$ such that $\beta-\eta \le n^{-1}\ip{\bs_i}{\bs_j}\le \beta$ for $1\le i<j\le m$. We establish
\begin{lemma}\label{lemma:counting}
For $m=O(1)$ as $n\to\infty$,
\[
    M(m,\beta,\eta)\le \exp_2\left(n+n(m-1)h_b\left(\frac{1-\beta+\eta}{2}\right)+O(\log_2 n)\right).
    \]
\end{lemma}
 Lemma~\ref{lemma:counting} is verbatim from~\cite[Lemma~6.7]{gamarnik2022algorithms}, we include the proof for completeness.
\begin{proof}[Proof of Lemma~\ref{lemma:counting}]
    Observe that $\ip{\bs}{\bs'}=n-2d_H(\bs,\bs')$ for any $\bs,\bs'\in\Sigma_n$. There are $2^n$ choices for $\bs_1$. Having fixed a $\bs_1$, there are
    \[
\sum_{\substack{\rho:\frac{1-\beta}{2}\le \rho\le \frac{1-\beta+\eta}{2} \\ \rho n\in\mathbb{N}}}\binom{n}{n\rho}\le \binom{n}{n\frac{1-\beta+\eta}{2}}n^{O(1)},
\]
choices for any  $\bs_i$, $2\le i\le m$, under the constraint $\beta-\eta\le n^{-1}\ip{\bs_1}{\bs_i}\le \beta$. Next, for any $\rho\in(0,1)$, $\binom{n}{n\rho}=\exp_2 \bigl(nh(\rho)+O\bigl(\log_2 n\bigr)\bigr)$ by Stirling's approximation. Combining these and recalling $m=O_n(1)$, we obtain Lemma~\ref{lemma:counting}.
\end{proof}
\paragraph{Probability estimate} 
Fix any $(\bs_1,\dots,\bs_m)$ with 
\[
\frac1n\ip{\bs_i}{\bs_j} = \beta-\eta_{ij},\quad 1\le i<j\le m.
\]
Clearly $0\le \eta_{ij}\le \eta$. Further, let $\boldsymbol{\eta} = (\eta_{ij}:1\le i<j\le m)\in\R^{m(m-1)/2}$. Then $\|\boldsymbol{\eta}\|_\infty\le \eta$. Our eventual choice of parameters $\beta,\eta$ and $m$ will ensure
\begin{equation}\label{eq:eta-bd}
    \eta=\frac{1-\beta}{2m}.
\end{equation}
We now control the probability term.
\begin{lemma}\label{lemma:prob-term}
    Let $\Sigma(\boldsymbol{\eta})\in\R^{m\times m}$ with unit diagonal entries such that for $1\le i<j\le m$,
    \[
\bigl(\Sigma(\boldsymbol{\eta})\bigr)_{ij}=\bigl(\Sigma(\boldsymbol{\eta})\bigr)_{ji}=\beta-\eta_{ij}.
    \]
    Then, the following holds.
    \begin{itemize}
        \item[(a)] $\Sigma(\boldsymbol{\eta})$ is positive definite (PD) if $\eta$ satisfies~\eqref{eq:eta-bd}. 
        \item[(b)] Suppose that $\eta$ satisfies~\eqref{eq:eta-bd}. Then, 
        \begin{align*}
           \mathbb{P}\left[\exists \tau_1,\dots,\tau_m\in\mathcal{I}:\max_{1\le i\le m}\bigl\|\M_i(\tau_i)\bs_i\bigr\|_\infty\le K\right] \le |\mathcal{I}|^m (2\pi)^{-\frac{mM}{2}}\left(\frac{1-\beta}{2}\right)^{-\frac{Mm}{2}}\left(\frac{2K}{\sqrt{n}}\right)^{Mm}.
        \end{align*}
    \end{itemize}
\end{lemma}
\begin{proof}[Proof of Lemma~\ref{lemma:prob-term}]
    \phantom{//}
    \paragraph{Part {\rm (a)}} Let $E\in\R^{m\times m}$ such that $E_{ii}=0$ and $E_{ij}=E_{ji}=-\eta_{ij}$ for $1\le i<j\le m$. Then,
    \[
    \Sigma(\boldsymbol{\eta})=(1-\beta)I + \beta \boldsymbol{1}\boldsymbol{1}^T +E.
    \]
    Note that the smallest eigenvalue of $(1-\beta)I + \beta \boldsymbol{1}\boldsymbol{1}^T$ is $1-\beta$ and $\|E\|_2\le \|E\|_F<\eta m$. So, $\Sigma(\boldsymbol{\eta})$ is invertible if $\eta<(1-\beta)/m$. Recalling the fact it is a covariance matrix, so in particular positive semidefinite, we establish part ${\rm (a)}$. 

    \paragraph{Part {\rm (b)}} As a first step, we take a union bound over $\mathcal{I}$ to obtain 
    \begin{equation}\label{eq:step1-ubd}
             \mathbb{P}\left[\exists \tau_1,\dots,\tau_m\in\mathcal{I}:\max_{1\le i\le m}\bigl\|\M_i(\tau_i)\bs_i\bigr\|_\infty\le K\right]\le |\mathcal{I}|^m \max_{\tau_i\in\mathcal{I},1\le i\le m}  \mathbb{P}\left[\max_{1\le i\le m}\bigl\|\M_i(\tau_i)\bs_i\bigr\|_\infty\le K\right]. 
        \end{equation}
        Next, denote by $R_i\sim \cN(0,I_n)$ the first row of $\M_i(\tau_i)\in\R^{M\times n}$, $1\le i\le m$. Observe that using the fact each $\M_i(\tau_i)$ has independent rows, 
        \begin{equation}\label{eq:step2-ind}
            \mathbb{P}\left[\max_{1\le i\le m}\bigl\|\M_i(\tau_i)\bs_i\bigr\|_\infty\le K\right]\le \mathbb{P}\left[\max_{1\le i\le m}n^{-\frac12}\left|\ip{R_i}{\bs_i}\right|\le \frac{K}{\sqrt{n}}\right]^M.
        \end{equation}
        Next, we consider the multivariate normal random vector $\bigl(n^{-1/2}\ip{R_i}{\bs_i}:1\le i\le m\bigr)$ consisting of standard normal coordinates. Let $\overline{\Sigma}$ denotes its covariance matrix, which depends on the choice of $\tau_1,\dots,\tau_m$. Observe that for $1\le i<j\le m$,
        \[
        \overline{\Sigma}_{ij} = \frac1n\mathbb{E}\bigl[\ip{R_i}{\bs_i}\ip{R_j}{\bs_j}\bigr]=\frac1n(\bs_i)^T \underbrace{\mathbb{E}[R_iR_j^T]}_{=\cos(\tau_i)\cos(\tau_j)I_m}\bs_j = \cos(\tau_i)\cos(\tau_j)(\beta-\eta_{ij}).
        \]
        We now remove the dependence on $\tau_i$ by relying on a Gaussian comparison inequality, due to Sid{\'a}k~\cite[Corollary~1]{sidak1968multivariate}. The version below is reproduced from~\cite[Theorem~6.5]{gamarnik2022algorithms}. 
\begin{theorem}\label{thm:sidak}
Let $(X_1,\dots,X_k)\in\R^k$ be a centered multivariate  normal random vector. Suppose that its covariance matrix $\Sigma\in\R^{k\times k}$ has unit diagonal entries has the following form: there exists $0\le \lambda_i\le 1$, $1\le i\le k$, such that for any $1\le i\ne j\le  k$, $\Sigma_{ij}=\lambda_i\lambda_j\rho_{ij}$ where $(\rho_{ij}:1\le i\ne j\le k)$ is a fixed arbitrary covariance matrix. Fix values $c_1,\dots,c_k>0$, and set
\[
P(\lambda_1,\dots,\lambda_k) = \mathbb{P}\bigl[|X_1|<c_1,|X_2|<c_2,\dots,|X_k|<c_k\bigr].
\]
Then, $P(\lambda_1,\dots,\lambda_k)$ is a non-decreasing function of each $\lambda_i$, $i=1,2,\dots,k$, $0\le \lambda_i\le 1$. That is,
\[
P(\lambda_1,\lambda_2,\dots,\lambda_k)\le P(1,1,\dots,1).
\]
\end{theorem}
We now let $(Z_1,\dots,Z_m)$ to be a centered multivariate normal random vector with covariance $\Sigma(\boldsymbol{\eta})$. Observe that
\begin{align}
\max_{\tau_1,\dots,\tau_m\in\mathcal{I}}\mathbb{P}\left[\max_{1\le i\le m}n^{-\frac12}\left|\ip{R_i}{\bs_i}\right|\le \frac{K}{\sqrt{n}}\right]&\le \mathbb{P}\left[\max_{1\le i\le m}|Z_i|\le \frac{K}{\sqrt{n}}\right]\label{eq:sidak} \\
& = (2\pi)^{-\frac{m}{2}}|\Sigma(\boldsymbol{\eta})|^{-\frac12}\int_{\boldsymbol{z}\in\left[-\frac{K}{\sqrt{n}},\frac{K}{\sqrt{n}}\right]^m}\exp\left(-\frac{\boldsymbol{z}^T \Sigma(\boldsymbol{\eta})^{-1} \boldsymbol{z}}{2}\right)\;d\boldsymbol{z} \nonumber \\
&\le (2\pi)^{-\frac{m}{2}}|\Sigma(\boldsymbol{\eta})|^{-\frac12} \left(\frac{2K}{\sqrt{n}}\right)^m\label{eq:density-at-most-1},
\end{align}
where~\eqref{eq:sidak} follows from Theorem~\ref{thm:sidak} and~\eqref{eq:density-at-most-1} follows from the trivial fact $\exp\left(-\frac{\boldsymbol{z}^T \Sigma(\boldsymbol{\eta})^{-1} \boldsymbol{z}}{2}\right)\le 1$. 

We lastly bound $|\Sigma(\boldsymbol{\eta})|$. For this, we rely on the following tool from matrix analysis.
\begin{theorem}[Hoffman-Wielandt Inequality]\label{thm:hw}
    Let $A\in\R^{m\times m}$ and $A+E\in\R^{m\times m}$ be two symmetric matrices with eigenvalues
    \[
    \lambda_1(A)\ge \cdots \ge \lambda_m(A) \quad\text{and}\quad \lambda_1(A+E)\ge \cdots \ge \lambda_m(A+E).
    \]
    Then,
    \[
    \sum_{1\le i\le m}\bigl(\lambda_i(A+E)-\lambda_i(A)\bigr)^2 \le \|E\|_F.
    \]
\end{theorem}
See~\cite[Corollary~6.3.8]{horn2012matrix} for a reference, and~\cite{hoffman1953variation} for the original paper.  We apply Theorem~\ref{thm:hw} to $\Sigma(\boldsymbol{\eta})$. Let $A=(1-\beta)I+\beta \boldsymbol{1}\boldsymbol{1}^T$ with eigenvalues $\lambda_1 = 1-\beta+\beta m>\lambda_2 = \cdots = \lambda_m = 1-\beta$ and $E$ be as above. Suppose that the eigenvalues of $A+E$ are $\mu_1\ge \cdots \ge \mu_m$. Fix any $2\le i\le m$. Theorem~\ref{thm:hw} yields
\[
\bigl|\mu_i-(1-\beta)\bigr|\le \|E\|_F \le \eta m =\frac{1-\beta}{2},
\]
yielding
\[
\mu_i\ge \frac{1-\beta}{2},\quad 2\le i\le m.
\]
Furthermore, this bounds extends to $\mu_1$, too, as
\[
\mu_1 \ge 1-\beta+\beta m - \frac{1-\beta}{2}>\frac{1-\beta}{2}.
\]
Since $\boldsymbol{\eta}\in \R^{m(m-1)/2}$ is arbitrary with $\|\boldsymbol{\eta}\|_\infty\le \eta \le \frac{1-\beta}{2m}$, we obtain  
\begin{equation}\label{eq:det-bd}
\inf_{\substack{\boldsymbol{\eta}\in\R^{m(m-1)/2} \\\|\boldsymbol{\eta}\|_\infty\le \frac{1-\beta}{2m}}}\bigl|\Sigma(\boldsymbol{\eta})\bigr| = \prod_{1\le i\le m}\mu_i \ge \left(\frac{1-\beta}{2}\right)^m.    
\end{equation}
Finally, combining~\eqref{eq:step1-ubd},~\eqref{eq:step2-ind},~\eqref{eq:density-at-most-1}, and~\eqref{eq:det-bd} we establish the proof of part {\rm (b)}.
\end{proof}
\paragraph{Estimating the expectation} Let $\mathcal{F}(m,\beta,\eta)$ be the set of all $m$-tuples $(\bs_1,\dots,\bs_m)$ such that $\beta-\eta\le n^{-1}\ip{\bs_i}{\bs_j}\le \beta$, $1\le i<j\le m$. Then
\[
\bigl|\mathcal{S}(K,m,\beta,\eta,\mathcal{I})\bigr| = \sum_{(\bs_1,\dots,\bs_m)\in\mathcal{F}(m,\beta,\eta)} \ind\left\{\exists \tau_1,\dots,\tau_m\in\mathcal{I}:\max_{1\le i\le m}\bigl\|\M_i(\tau_i)\bs_i\bigr\|_\infty\le K\right\}.
\]
Using linearity of expectation, Lemma~\ref{lemma:counting}, Lemma~\ref{lemma:prob-term}, and the fact $\log_2|\mathcal{I}|\le cn$, we obtain
\begin{align}
    \mathbb{E}\bigl[\bigl|\mathcal{S}(K,m,\beta,\eta,\mathcal{I})\bigr|\bigr]&\le \exp_2\Bigl(\Psi(m,\beta,\eta,c)+O(\log_2 n)\Bigr),\label{eq:first-mome}
\end{align}
where
\begin{equation}\label{eq:free-energy}
    \Psi(m,\beta,\eta,c) = n+mnh_b\left(\frac{1-\beta+\eta}{2}\right)+cmn +\frac{mM}{2}\log_2\frac{4K^2}{\pi(1-\beta)} -\frac{Mm}{2}\log_2 n.
\end{equation}
We set $\eta$ and $c$ as
\begin{equation}\label{eq:eta-and-c}
    \eta = \frac{1-\beta}{2m} \quad\text{and}\quad c=\frac{1}{m},
\end{equation}
where we recalled $\eta$ from~\eqref{eq:eta-bd}; parameters $\beta$ and $m$ are to be tuned soon. We now recall the scaling on $n$ from~\eqref{eq:n-scale}. In particular, 
\[
\log_2 n\ge \log_2 c_2 + \log_2 M+\log_2 \log_2 M\ge \log_2 c_2+ \log_2 M.
\]
With this, we arrive at
\begin{align}
   \Psi\left(m,\beta,\frac{1-\beta}{2m},\frac1m\right)
   &\le 2C_1 M\log_2 M+mC_1 M\log_2 M \cdot h_b\left(\frac{1-\beta}{2}+\frac{1-\beta}{4m}\right) \nonumber \\
   &+\frac{mM}{2}\log_2 \frac{4K^2}{\pi(1-\beta)c_2}-\frac{Mm}{2}\log_2 M\label{eq:pu1}.
\end{align}
Note that if $\beta\in(1/2,1)$ and $m\in\mathbb{N}$, we clearly have
\[
h_b\left(\frac{1-\beta}{2}+\frac{1-\beta}{4m}\right) \le h_b(1-\beta).
\]
We choose $\beta^*>1/2$ such that 
\[
h_b(1-\beta^*) =\min\left\{\frac{1}{4C_1},\frac12\right\}.
\]
So, 
\begin{equation}\label{eq:entropy-term-temp}
   m C_1 M\log_2 M \cdot h_b\left(\frac{1-\beta^*}{2}+\frac{1-\beta^*}{4m}\right)\le \frac{Mm}{4}\log_2 M.
\end{equation}
Combining~\eqref{eq:pu1} and~\eqref{eq:entropy-term-temp}, we further upper bound
\begin{equation}\label{eq:free-en}
    \Psi\left(m,\beta^*,\frac{1-\beta^*}{2m},\frac1m\right)\le 2C_1 M\log_2 M -\frac{Mm}{4}\log_2 M +\Theta(mM).
\end{equation}
Finally, taking $m=m^*=\max\{2,16C_1\}$, we get
\begin{equation}\label{eq:free-en2}
 \Psi\left(m^*,\beta^*,\frac{1-\beta^*}{2m^*},\frac{1}{m^*}\right) = -\Theta(M\log_2 M).
\end{equation}
Combining~\eqref{eq:first-mome} with the fact $O(\log_2 n)=O(\log_2 M)=o(M\log_2 M)$ as $M=\omega(1)$, we conclude that
\[
\mathbb{E}\bigl[\bigl|\mathcal{S}(K,m^*,\beta^*,\eta^*,\mathcal{I})\bigr|\bigr]\le \exp_2\left(\Psi\left(m^*,\beta^*,\frac{1-\beta^*}{2m^*},\frac{1}{m^*}\right)+o(M\log_2 M)\right) = 2^{-\Theta(M\log M)} = 2^{-\Theta(n)}.
\]
This completes the proof of Theorem~\ref{thm:m-ogp-discrepancy}.
\subsubsection*{Acknowledgments}
The first author is supported in part by NSF grant DMS-2015517. The second author is supported by a Columbia University, Distinguished Postdoctoral Fellowship in Statistics. The third author is supported in part by NSF grant DMS-1847451.
\bibliographystyle{amsalpha}
\bibliography{bibliography}

\end{document}